\documentclass[11pt]{article}

\usepackage{amssymb} 
\usepackage{amsmath}     
\usepackage{amsthm}
\usepackage{todonotes}
\usepackage{mathtools}
\usepackage{authblk}

\usepackage{algorithm, algorithmic}

\usepackage{a4wide}
\usepackage{graphicx}
\usepackage{hyperref}
\newcommand{\Rset}{\mathbb{R}}

\newcommand{\cU}{\mathcal{U}}

\newcommand{\cUhp}{\mathcal{U}^{\mathcal{H\!P}}}
\newcommand{\cUvp}{\mathcal{U}^{\mathcal{V\!P}}}
\newcommand{\cUe}{\mathcal{U}^{\mathcal{E}}}

\newtheorem{thm}{Theorem}
\newtheorem{lem}{Lemma}

\newtheorem{cor}{Corollary}

\newtheorem{obs}{Observation}

\newtheorem{pro}{Property}
\newtheorem{cla}{Claim}

\begin{document}


\title{Robust two-stage combinatorial optimization problems under convex uncertainty}

\author[1]{Marc Goerigk}
\author[2]{Adam Kasperski}
\author[3]{Pawe{\l} Zieli\'nski}

\affil[1]{Network and Data Science Management, University of Siegen, Germany\\
           \texttt{marc.goerigk@uni-siegen.de}}
\affil[2]{Faculty of Computer Science and Management, 
Wroc{\l}aw  University of Science and Technology, Poland\\
            \texttt{adam.kasperski@pwr.edu.pl}}
\affil[3]{Faculty of Fundamental Problems of Technology, 
Wroc{\l}aw  University of Science and Technology, Poland\\
       \texttt{pawel.zielinski@pwr.edu.pl}}
   
    \date{}
    
\maketitle

 \begin{abstract}
	In this paper a class of robust two-stage combinatorial optimization problems is discussed.	
	 It is assumed that the uncertain second stage costs are specified in the form of a convex uncertainty set, in particular polyhedral or ellipsoidal ones. It is shown that the robust two-stage versions of basic network and selection problems are NP-hard, even in a very restrictive cases. Some exact and approximation algorithms for the general problem are constructed. Polynomial and approximation algorithms for the robust two-stage versions of basic problems, such as the selection and shortest path
	 problems,  are also provided.
 \end{abstract}
 

\section{Introduction}

In a traditional combinatorial optimization problem we seek a cheapest object composed of elements chosen from a finite element set $E$. For example, $E$ can be a set of arcs of a given graph with specified arc costs, and we wish to compute an $s-t$ path, spanning tree, perfect matching etc.\ with  minimum costs (see, for example,~\cite{AMO93,PS98}). In many practical situations the exact values of the element costs are unknown. An uncertainty (scenario) set $\mathcal{U}$ is then provided, which contains all realizations of the element costs, called scenarios, which may occur. The probability distribution in $\mathcal{U}$ can be known, partially known, or unknown. In the latter case the robust optimization framework can be used, which consists in computing a solution minimizing the cost in a worst case. 
Single-stage robust combinatorial optimization problems, under various uncertainty sets, have been extensively discussed over the last decade. Survey of the results in this area can be found in~\cite{ABV09, KZ16b, GS15, BK18}. For these problems a complete solution must be determined  before the true scenario is revealed.

In many practical applications a solution can be constructed in more than one stage. For combinatorial problems, a part of the object can be chosen now (in the first stage) and completed in a future (in the second stage), after the structure of the costs has been changed. Typically, the first stage costs are known while the second stage costs can only be predicted to belong to an uncertainty set $\mathcal{U}$. First such models were discussed in~\cite{DRM05, FFK06, KMU08, KZ11}, where the robust two-stage spanning tree and perfect matching problems were considered. In these papers, the uncertainty set $\mathcal{U}$ contains $K$ explicitly listed scenarios. Several negative and positive complexity results for this uncertainty representation were established. Some of them have been recently extended in~\cite{GKZ18}, where also the robust two-stage shortest path problem has been investigated. In~\cite{KZ15b} and~\cite{CGKZ18} the robust two-stage selection problem has been explored.
The problem is NP-hard for discrete uncertainty representation but it is polynomially solvable under a special case of polyhedral uncertainty set, called continuous budgeted uncertainty (see~\cite{CGKZ18}). 

Robust two-stage problems belong to the class of three-level, min-max-min optimization problems. In mathematical programming, this approach is also called \emph{adjustable robustness} (see, e.g.~\cite{BTG04,yanikouglu2018survey}). Namely, some variables must be determined before the realization of the uncertain parameters, while the other part are variables that can be chosen after the realization. Several such models have been recently considered in combinatorial optimization, which can be represented as a 0-1 programming problem. Among them there is the robust two-stage problem discussed in this paper, but also the robust recoverable models~\cite{B11,B12} and the $k$-adaptability approach~\cite{BK17}. In general, problems of this type can be hard to solve exactly. A standard approach is to apply row and column generation techniques, which consists in solving a sequence of MIP formulations (see, e.g.,~\cite{ZZ13}). However, this method can be inefficient for larger problems, especially when the underlying deterministic problem is already NP-hard. Therefore, some faster approximation algorithms can be useful in this case.

In this paper we consider the class of robust two-stage combinatorial problems under convex uncertainty, i.e. when the uncertainty set $\mathcal{U}$ is convex. Important special cases are polyhedral and ellipsoidal uncertainty, which are widely used in single-stage robust optimization. Notice that in the problems discussed in~\cite{DRM05, FFK06, KMU08,KZ11}, $\mathcal{U}$ contains a fixed number of scenarios, so it is not a convex set.
 The problem formulation and description of the uncertainty sets are provided in Section~\ref{secform}.  The complexity status of basic problems, in particular network and selection problems, has been open to date. In Section~\ref{seccomplex} we show that all these basic problems are NP-hard, both under polyhedral and ellipsoidal uncertainty.  In Section~\ref{secmip}, we construct compact MIP formulations for a special class of robust two-stage combinatorial problems and show several of its properties.
  In Section~\ref{apprgen}, we propose  an algorithm  for the general problem, which returns an approximate solution with some guaranteed worst case ratio. This algorithm does not run in polynomial time. However, it requires solving only one (possibly NP-hard)  MIP formulation, while a compact MIP formulation for the general case is unknown.
Finally, in Sections~\ref{secsel},~\ref{secrs}, and~\ref{secsp} we study the robust two-stage versions of three particular problems, namely the selection, representatives selection and shortest path ones. We show some additional negative and positive complexity results for them.  There is still a number of open questions concerning the robust two-stage approach. We will state them in the last section.

\section{Problem formulation}
\label{secform}

Consider the following generic combinatorial optimization problem~$\mathcal{P}$:
$$\begin{array}{llll}
		\min & \pmb{C}^T\pmb{x} \\
			& \pmb{x}\in \mathcal{X} \subseteq \{0,1\}^n,\\
	\end{array}
	$$
where $\pmb{C}=[C_1,\dots,C_n]^T$ is a vector of nonnegative costs and $\mathcal{X}$ is a set of feasible solutions. In this paper we consider the general problem~$\mathcal{P}$, as well as the following special cases:
\begin{enumerate}
	\item Let $G=(V,A)$ be a given network, where $C_i$ is a cost of arc $a_i\in A$. Set $\mathcal{X}$ contains characteristic vectors of some objects in $G$, for example the simple $s-t$ paths or spanning trees. Hence $\mathcal{P}$ is the \textsc{Shortest Path} or \textsc{Spanning Tree} problem, respectively. These basic network problems are polynomially solvable, 
	see, e.g.,~\cite{AMO93, PS98}.
	\item Let $E=\{e_1,\ldots,e_n\}$ be a set of items. Each item $e_i\in E$ has a cost $C_i$ and 
	we wish to choose exactly $p$~items 
	out of set~$E$  to minimize the total cost.  Set $\mathcal{X}$ contains characteristic vectors of the feasible selections, i.e.
	$\mathcal{X}=\{\pmb{x}\in\{0,1\}^n: \sum_{i\in [n]} x_i = p\}$. We will denote by $[n]$ the set $\{1,\dots,n\}$.
	This is the \textsc{Selection} problem whose robust single and two-stage versions were discussed in~\cite{AV01, C04,KZ15b, CGKZ18}.
	\item 
	Let $E=\{e_1,\ldots,e_n\}$ be a set of tools (items).
	This set is partitioned into a family of disjoint sets $T_l$, $l\in [\ell]$.
	Each tool $e_i\in E$ has a cost $C_i$ and 
	we wish to select exactly one tool from each subset~$T_l$ to minimize  their   total cost.
	Set $\mathcal{X}$ contains characteristic vectors of the feasible selections, i.e.
	$\mathcal{X}=\{\pmb{x}\in \{0,1\}^n: \sum_{i\in T_l} x_i= 1, l\in [\ell]\}$. 
	This is the \textsc{Representatives Selection} problem (\textsc{RS} for short) 
	whose robust single-stage version was considered in~\cite{DK12, DW13, KKZ15}.
\end{enumerate}
Given a vector $\pmb{x}\in \{0,1\}^n$, let us define the following set of \emph{recourse actions}:
$$\mathcal{R}(\pmb{x})=\{\pmb{y}\in \{0,1\}^n: \pmb{x}+\pmb{y}\in \mathcal{X}\}$$
and a set of \emph{partial solutions} is defined as follows:
$$\mathcal{X}'=\{\pmb{x}\in \{0,1\}^n: \mathcal{R}(\pmb{x})\neq \emptyset\}.$$
Observe that $\mathcal{X}\subseteq \mathcal{X}'$ and $\mathcal{X}'$ contains all vectors which can be completed to a feasible solution in $\mathcal{X}$.
A partial solution $\pmb{x}\in \mathcal{X}'$ is  completed in the second stage, i.e. we choose $\pmb{y}\in \mathcal{R}(\pmb{x})$ which yields $(\pmb{x}+\pmb{y})\in \mathcal{X}$. The overall cost of the solution constructed is $\pmb{C}^T\pmb{x}+\pmb{c}^T\pmb{y}$ for a fixed second-stage cost vector $\pmb{c}=[c_1,\ldots,c_n]^T$. We assume that the vector of the first-stage costs $\pmb{C}$ is known but the vector of the second-stage costs 
is uncertain and belongs to a specified uncertainty (scenario) set $\mathcal{U}\subset \Rset^n_{+}$. In this paper, we discuss the following \emph{robust two-stage problem}:
$$
	\textsc{RTSt}:\; \min_{\pmb{x}\in\mathcal{X}'}\max_{\pmb{c}\in \mathcal{U}}\min_{\pmb{y}\in \mathcal{R}(\pmb{x})}(\pmb{C}^T\pmb{x}+\pmb{c}^T\pmb{y}).
$$
The \textsc{RTSt} problem is a robust two-stage version of the problem~$\mathcal{P}$.
It is worth pointing out that  \textsc{RTSt}  is a generalization  of four problems, which we also examine in this paper.
Namely,
given $\pmb{x}\in \mathcal{X}'$ and $\pmb{c}\in\mathcal{U}$, we consider the following \emph{incremental problem:}
$$\textsc{Inc}(\pmb{x},\pmb{c})=\min_{\pmb{y}\in \mathcal{R}(\pmb{x})} \pmb{c}^T\pmb{y}.$$
Given scenario $\pmb{c}\in \mathcal{U}$, we  study the following \emph{two-stage} problem:
$$\textsc{TSt}(\pmb{c})=\min_{\pmb{x}\in \mathcal{X}'}\min_{\pmb{y}\in \mathcal{R}(\pmb{x})} (\pmb{C}^T\pmb{x}+\pmb{c}^T\pmb{y}).$$
Finally, given $\pmb{x}\in \mathcal{X}'$, we also consider the following \emph{evaluation} problem: 
$$\textsc{Eval}(\pmb{x})=\pmb{C}^T\pmb{x}+\max_{\pmb{c}\in \mathcal{U}}\min_{\pmb{y}\in \mathcal{R}(\pmb{x})} \pmb{c}^T\pmb{y}=\pmb{C}^T\pmb{x}+\max_{\pmb{c}\in\mathcal{U}}\textsc{Inc}(\pmb{x},\pmb{c}).$$
A scenario $\pmb{c}$ which maximizes $\textsc{Inc}(\pmb{x},\pmb{c})$ is called a \emph{worst scenario} for $\pmb{x}$. The inner maximization problem is called the \emph{adversarial problem}, i.e., the problem
\[ \max_{\pmb{c}\in \mathcal{U}}\min_{\pmb{y}\in \mathcal{R}(\pmb{x})}(\pmb{C}^T\pmb{x}+\pmb{c}^T\pmb{y}) \]
Notice that the robust two stage problem can be equivalently represented as follows:
$$\textsc{RTSt}:\; \min_{\pmb{x}\in \mathcal{X}'} \textsc{Eval}(\pmb{x}).$$
Further notice that the two-stage problem is a special case of \textsc{RTSt}, where $\mathcal{U}=\{\pmb{c}\}$ contains only one scenario. The following fact is exploited later in this paper:
\begin{obs}
\label{obsones}
        Computing $\textsc{TSt}(\pmb{c})$ for a given $\pmb{c}\in \mathcal{U}$ (solving the two-stage problem)
        boils down to solving the underlying deterministic problem~$\mathcal{P}$.
\end{obs}
\begin{proof}
	Let $\hat{c}_i=\min\{C_i, c_i\}$ for each $i\in [n]$ and let $\hat{\pmb{z}}$ be an optimal solution to problem $\mathcal{P}$ for the costs $\hat{\pmb{c}}$. Consider solution $(\hat{\pmb{x}},\hat{\pmb{y}})$ constructed as follows: 
	set  $\hat{x}_i=0$, $\hat{y}_i=0$ if $\hat{z}_i=0$;
	set  $\hat{x}_i=1$, $\hat{y}_i=0$ if $\hat{z}_i=1$ and $\hat{c}_i=C_i$; 
	set  $\hat{x}_i=0$, $\hat{y}_i=1$ if $\hat{z}_i=1$ and $\hat{c}_i=c_i$.
	Of course, $\hat{\pmb{x}}\in \mathcal{X}'$ and $\hat{\pmb{y}}\in \mathcal{R}(\hat{\pmb{x}})$. It is easy to verify that $(\hat{\pmb{x}},\hat{\pmb{y}})$ is an optimal solution to the two-stage problem with the objective value of~$\textsc{TSt}(\pmb{c})$.
\end{proof}

In this paper, we examine the following three types of convex uncertainty sets:
\begin{align}
\cUhp &=\{\underline{\pmb{c}}+\pmb{\delta}: \pmb{A}\pmb{\delta}\leq \pmb{b},\pmb{\delta}\geq \pmb{0}\}\subset \Rset^n_{+}, \\
\cUvp &={\rm conv}\{\pmb{c}_1,\ldots,\pmb{c}_K\}\subset \Rset^n_{+}, \\
\cUe &=\{\underline{\pmb{c}}+\pmb{A\delta}: ||\pmb{\delta}||_{2}\leq 1\}\subset \Rset^n_{+},
\end{align}
where 
$\underline{\pmb{c}}=[\underline{c}_1,\dots,\underline{c}_n]^T$ is the vector of nominal second stage costs, $\pmb{\delta}=[\delta_1,\dots,\delta_n]^T$ represents deviations of the second stage costs from their nominal values
and $\pmb{A}\in \Rset^{m\times n}$ is the deviation constraint matrix.
There is no loss of generality in assuming
 that all the sets are bounded. The uncertainty sets $\cUhp$ and $\cUvp$ are two representations of the \emph{polyhedral uncertainty}. By the decomposition theorem~\cite[Chapter 7.2]{SH98}, both representations are equivalent, i.e. bounded $\cUhp$ can be represented as $\cUvp$ and vice versa. However, the corresponding transformations need not be polynomial. Thus the complexity results from one type of polytope do not carry over to the other, and we consider them separately. The set $\cUe$  represents \emph{ellipsoidal uncertainty}, which is a popular uncertainty representation in robust optimization~(see, e.g.,~\cite{BN09}). We also study the following special cases of $\cUhp$:
\begin{align*}
\cUhp_0 &=\{\underline{\pmb{c}}+\pmb{\delta}: \pmb{0}\leq \pmb{\delta}\leq \pmb{d},||\pmb{\delta}||_1\leq \Gamma\}, \\
\cUhp_1 &= \{\underline{\pmb{c}}+\pmb{\delta}: \sum_{i\in U_j} \delta_i \leq \Gamma_j, j\in [K], \pmb{\delta}\geq \pmb{0}\}
\end{align*}

Set $\cUhp_0$ is called \emph{continuous budgeted uncertainty}~\cite{NO13, CGKZ18} and can be seen as a continuous and convex version of the nonconvex uncertainty set proposed in~\cite{BS04}. In set $\cUhp_1$ we have $K$ budget constraints defined for some (not necessarily disjoint) subsets $U_1,\dots,U_K\subseteq [n]$.

\section{General hardness results}
\label{seccomplex}

The robust two-stage problem is not easier than the underlying deterministic problem $\mathcal{P}$. So, it is interesting to characterize the complexity of \textsc{RTSt} when $\mathcal{P}$ is polynomially solvable. In this section we focus on a core problem, which is a special case of all the particular problems studied in Section~\ref{secform}. We will show that it is NP-hard under $\cUvp$, $\cUhp$ and $\cUe$. Hence we get hardness results for all the particular problems. Consider the following set of feasible solutions 
$$\mathcal{X}_{\pmb{1}}=\{\pmb{x}\in \{0,1\}^n: x_1+\dots+x_n= n\}=\{\pmb{1}\},$$
 i.e. $\mathcal{X}_{\pmb{1}}$ contains only the vector of ones.  We have $\mathcal{X}_{\pmb{1}}'=\{\pmb{x}\in\{0,1\}^n: x_1+\dots+x_n\leq n\}$ and $\mathcal{R}(\pmb{x})=\{\pmb{1}-\pmb{x}\}$ contains only one solution, as there is only one recourse action for each $\pmb{x}\in \mathcal{X}_{\pmb{1}}'$.  Hence, the robust two stage version of the problem with $\mathcal{X}_{\pmb{1}}$ can be rewritten as follows:
\begin{equation}
\label{rtstd}
	\textsc{RTSt}_{\pmb{1}}: \min_{\pmb{x}\in \mathcal{X}_1'} \left( \pmb{C}^T\pmb{x}+ \max_{\pmb{c}\in \mathcal{U}}\pmb{c}^T(\pmb{1}-\pmb{x}) \right).
\end{equation}
The following result is known:
\begin{thm}[\cite{KZ15b, GKZ18}]
\label{thmcopl1}
	The $\textsc{RTSt}_{\pmb{1}}$ problem with $\mathcal{U}=\{\pmb{c}_1,\pmb{c}_2\}{\subset{\mathbb{R}}^n_+}$ is NP-hard. Furthermore, if $\mathcal{U}=\{\pmb{c}_1,\dots, \pmb{c}_K\}{\subset{\mathbb{R}}^n_+}$ and $K$ is a part of the input, then $\textsc{RTSt}_{\pmb{1}}$ is strongly NP-hard.
\end{thm}

 We use Theorem~\ref{thmcopl1} to prove the next complexity results. First observe that the problem
 under consideration
  will not change if we replace $\mathcal{U}=\{\pmb{c}_1,\dots,\pmb{c}_K\}$ with $\cUvp={\rm conv}\{\pmb{c}_1,\dots, \pmb{c}_K\}$ in~(\ref{rtstd}). Hence, we immediately get the following corollary:
 
 \begin{cor}
 	The $\textsc{RTSt}_{\pmb{1}}$ problem with uncertainty set $\cUvp$ is NP-hard  when $K=2$ and strongly NP-hard when $K$ is a part of the input.
 \end{cor}

 \begin{thm}
 \label{corhp}
 	 The $\textsc{RTSt}_{\pmb{1}}$ problem with uncertainty set $\cUhp$ is strongly NP-hard.
 \end{thm}
 \begin{proof}
Let $\mathcal{I}=(n,\pmb{C},\mathcal{U}=\{\pmb{c}_1,\dots, \pmb{c}_K\})$, be an instance of the strongly NP-hard $\textsc{RTSt}_{\pmb{1}}$ problem. 
Consider an instance $\mathcal{I}_1=(n+K,[\pmb{\pmb{C}}, \pmb{0}]^T,\cUhp)$ of $\textsc{RTSt}_{\pmb{1}}$, where $[\pmb{\pmb{C}}, \pmb{0}]^T\in \Rset^{n+K}$ are the first stage costs and
\[
\cUhp = \left\{ \pmb{0} + \begin{bmatrix}\pmb{\pmb{\delta}}\\ \pmb{\lambda}\end{bmatrix}\ :\  \pmb{\delta}= \sum_{j\in[K]}  \lambda_j \pmb{c}_j,
\sum_{j\in[K]} \lambda_j = 1,
\delta_i \ge 0 \ \forall i\in[n],
\lambda_j \ge 0 \ \forall j\in[K] \right\} \subset \mathbb{R}^{n+K}.
\]
Since the first stage costs of variables $x_{n+1},\dots,x_{n+K}$ are~0, we can fix $x_{n+1}=\dots=x_{n+K}=1$ in every optimal solution to the instance $\mathcal{I}_1$. The problem then reduces to
$$\min_{\pmb{x}\in\mathcal{X}_{\pmb{1}}'}\max_{\{\pmb{\lambda}\geq \pmb{0}:|| \pmb{\lambda} ||_{1}=1\}} \left( \pmb{C}^T\pmb{x} + \sum_{j\in[K]}  \lambda_j \pmb{c}^T_j(\pmb{1}-\pmb{x}) \right)
=\min_{\pmb{x}\in\mathcal{X}_{\pmb{1}}'}\max_{\pmb{c}\in\{\pmb{c}_1,\dots,\pmb{c}_K\}}  \left( \pmb{C}^T\pmb{x} + \pmb{c}^T(\pmb{1}-\pmb{x}) \right),$$
where $\mathcal{X}_{\pmb{1}}'=\{\pmb{x}\in \{0,1\}^n: x_1+\dots+x_n\leq n\}$. Consequently, the problem with instance $\mathcal{I}_1$ is equivalent to the strongly NP-hard problem with the instance $\mathcal{I}$.
 \end{proof}
Note that the reduction in the proof of Theorem~\ref{corhp} constructs an uncertainty set $\cUhp$ with a non-constant number of constraints. We will show in Section~\ref{secmip} that if the number of constraints in the description of $\cUhp$ (except for the nonnegativity constraints) is constant, then the problem is polynomially solvable.

\begin{thm}
\label{thmcomplE}
	The $\textsc{RTSt}_{\pmb{1}}$ problem with uncertainty set $\cUe$ is NP-hard.
\end{thm}
\begin{proof}
	Given an instance $\mathcal{I}=(n,\pmb{C},\mathcal{U}=\{\pmb{c}_1,\pmb{c}_2\})$  of $\textsc{RTSt}_{\pmb{1}}$, define $\underline{\pmb{c}}=\pmb{c}_1+\pmb{c}_2$ and $\pmb{y}=\pmb{1}-\pmb{x}$.
	 We use the following equality (see~\cite{BS04a}):
$$2\cdot \max\{\pmb{c}_1^T\pmb{y},\pmb{c}_2^T\pmb{y}\}=(\pmb{c}_1^T\pmb{y}+\pmb{c}_2^T\pmb{y})+\sqrt{\pmb{y}^T(\pmb{c}_1-\pmb{c}_2)(\pmb{c}_1-\pmb{c}_2)^T\pmb{y}}=(\pmb{c}_1^T\pmb{y}+\pmb{c}^T_2\pmb{y})+\sqrt{\pmb{y}^T\pmb{A}\pmb{A}^T\pmb{y}},$$
where $\pmb{A}=[\pmb{c}_1-\pmb{c}_2,\pmb{0},\dots,\pmb{0}]$ is a square $n\times n$ matrix (we append $n-1$ columns $\pmb{0}\in \Rset^n$ to $\pmb{c}_1-\pmb{c}_2\in \Rset^n$.  We  get
\begin{align*}
2\cdot &\min_{\pmb{x}\in \mathcal{X}_{\pmb{1}}'}(\pmb{C}^T\pmb{x}+\max\{\pmb{c}_1^T\pmb{y},\pmb{c}_2^T\pmb{y}\})= \min_{\pmb{x}\in \mathcal{X}_{\pmb{1}}'} (2\pmb{C}^T\pmb{x}+\underline{\pmb{c}}^T\pmb{y}+\sqrt{\pmb{y}^T\pmb{A}\pmb{A}^T\pmb{y}}) \\
= &\min_{\pmb{x}\in \mathcal{X}_{\pmb{1}}'}(2\pmb{C}^T\pmb{x}+\underline{\pmb{c}}^T\pmb{y}+||\pmb{A}^T\pmb{y}||_{2})= \min_{\pmb{x}\in\mathcal{X}_{\pmb{1}}'} (2\pmb{C}^T\pmb{x}+ \max_{\pmb{c}\in \{\underline{\pmb{c}}+\pmb{A\delta}: ||\pmb{\delta}||_{2}\leq 1\}} \pmb{c}^T\pmb{y}).
\end{align*}
The last equality follows from the fact that 
$\underline{\pmb{c}}^T\pmb{y}+||\pmb{A}^T\pmb{y}||_{2}=\max_{\pmb{c}\in \{\underline{\pmb{c}}+\pmb{A\delta}: ||\pmb{\delta}||_{2}\leq 1\}} \pmb{c}^T\pmb{y}$ 
(see, e.g.,~\cite{BN99}).
In consequence, the NP-hard problem with the instance $\mathcal{I}$ is equivalent to $\textsc{RTSt}_{\pmb{1}}$ with the first stage costs $2\pmb{C}$ and ellipsoidal uncertainty set $\cUe=\{\underline{\pmb{c}}+\pmb{A\delta}: ||\pmb{\delta}||_{2}\leq 1\}$.
\end{proof}

\begin{thm}
\label{thmcompl}
	The robust two-stage versions of the \textsc{Selection}, \textsc{RS}, \textsc{Spanning Tree}, and \textsc{Shortest Path} problems are strongly NP-hard under $\cUvp$ and $\cUhp$, and NP-hard under $\cUe$.
\end{thm}
\begin{proof}
	It is easy to see that $\textsc{RTSt}_{\pmb{1}}$ is a special case of the \textsc{RTSt Selection} problem, with $p=n$,  and the \textsc{RTSt RS} problem, with $T_i=\{e_i\}$, $i\in [n]$. To see that it is also a special case of the basic network problems, consider the (chain) network $G=(V,A)$ shown in Figure~\ref{figchain}. This network contains exactly one $s-t$ path and spanning tree. So the problem is only to decide for each arc, whether to choose it in the first or in the second stage, which is equivalent to solving $\textsc{RTSt}_{\pmb{1}}$.	
	\begin{figure}[ht]
		\centering
		\includegraphics{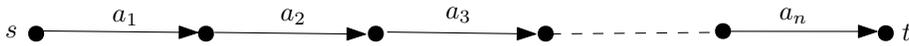}
		\caption{Illustration of the proof of Theorem~\ref{thmcompl}.}\label{figchain}
	\end{figure}
\end{proof}
In Section~\ref{secsp} we will show that the hardness result from Theorem~\ref{thmcompl} can be strengthened for the two-stage version of the \textsc{Shortest Path} problem.

\section{Compact formulations}
\label{secmip}

In this section we construct compact formulations for a special class of problems under uncertainty sets $\cUhp$ and $\cUe$. We will assume that 
\begin{equation}
\label{eqx}
\mathcal{X}=\{\pmb{x}\in \{0,1\}^n: \pmb{H}\pmb{x}\geq \pmb{g}\}
\end{equation}
and the polyhedron 
\begin{equation}
\mathcal{N}=\{\pmb{x}\in \Rset^n: \pmb{H}\pmb{x}\geq \pmb{g}, \pmb{0}\leq \pmb{x} \leq \pmb{1}\}
\label{polh}
\end{equation}
 is \emph{integral}, i.e.  $\mathcal{N}$ is the convex hull of all integral vectors
in~$\mathcal{N}$ or, equivalently, $\min(\max)\{\pmb{c}^T\pmb{x}: \pmb{x}\in \mathcal{N}\}$ is
attained by an integral vector, for each~$\pmb{c}$ for which the minimum (maximum) is finite
(see~\cite[Chapter~16.3]{SH98}).
Important examples, where the set of feasible solutions is described by $\mathcal{N}$ are 
the shortest path and the selection problems discussed in Section~\ref{secform}. 
 We can also use the constraints $\pmb{H}\pmb{x}=\pmb{g}$ to describe $\mathcal{X}$ and the further reasoning will be the same.
 We can rewrite the inner adversarial problem (notice that $\pmb{x}\in\{0,1\}^n$ is fixed) as follows:
\begin{align*}
&\max_{\pmb{c}\in \mathcal{U}}\min_{\pmb{y}\in \mathcal{R}(\pmb{x})} \pmb{c}^T\pmb{y} \\
= &\max_{\pmb{c}\in \mathcal{U}}\min_{\{\pmb{y}\in\{0,1\}^n:\pmb{y}+\pmb{x}\in \mathcal{X}\}} \pmb{c}^T\pmb{y} \\
= &\max_{\pmb{c}\in \mathcal{U}}\min_{\{\pmb{y}\in\{0,1\}^n:\pmb{H}(\pmb{y}+\pmb{x})\geq \pmb{g},\; \pmb{y}\leq \pmb{1}-\pmb{x}\}} \pmb{c}^T\pmb{y}\\
= &\max_{\pmb{c}\in \mathcal{U}}\min_{\{\pmb{y}\in\mathbb{R}^n:\pmb{H}(\pmb{y}+\pmb{x})\geq \pmb{g},\; \pmb{0}\leq \pmb{y}\leq \pmb{1}-\pmb{x}\}} \pmb{c}^T\pmb{y},
\end{align*}
where the last equality follows from the integrality assumptions and the fact that $\pmb{x}$ is a fixed binary vector.  Since $\mathcal{U}$ and $\{\pmb{y}:\pmb{H}(\pmb{y}+\pmb{x})\geq \pmb{g},\; \pmb{0}\leq \pmb{y}\leq \pmb{1}-\pmb{x}\}$ are convex (compact) sets
and $\pmb{c}^T\pmb{y}$ is a concave-convex function,
by the minimax theorem~\cite{N28}  we can rewrite the adversarial problem as follows:
\begin{equation}
\label{mip1}
	\min_{\{\pmb{y}\in\mathbb{R}^n:\pmb{H}(\pmb{y}+\pmb{x})\geq \pmb{g},\; \pmb{0}\leq \pmb{y}\leq \pmb{1}-\pmb{x}\}} \max_{\pmb{c}\in \mathcal{U}} \pmb{c}^T\pmb{y}.
\end{equation}
The robust two-stage problem thus becomes the following min-max problem:
\begin{equation}
\label{mip2}
\min_{\pmb{x}\in\mathcal{X}'} \min_{\{\pmb{y}\in\mathbb{R}^n:\pmb{H}(\pmb{y}+\pmb{x})\geq \pmb{g},\; \pmb{0}\leq \pmb{y}\leq \pmb{1}-\pmb{x}\}} \max_{\pmb{c}\in \mathcal{U}} \left( \pmb{C}^T \pmb{x} + \pmb{c}^T\pmb{y}\right).
\end{equation}


If $\mathcal{U}=\cUhp$, then we can dualize the inner maximization problem in~(\ref{mip2}), obtaining 
$$ \max_{\pmb{c}\in \mathcal{U}}\pmb{c}^T\pmb{y}=\max_{\{\pmb{\delta}\geq \pmb{0}: \pmb{A\delta}\leq \pmb{b}\}} (\underline{\pmb{c}}+\pmb{\delta})^T\pmb{y}=\underline{\pmb{c}}^T\pmb{y}+\min_{\{\pmb{u}\geq \pmb{0}: \pmb{u}^T\pmb{A}\geq \pmb{y}^T\}} \pmb{u}^T\pmb{b}.$$
As the result we get the following compact MIP formulation for \textsc{RTSt} under $\cUhp$:
\begin{equation}
\label{mip2st}
	\begin{array}{lllll}
		\min & \pmb{C}^T\pmb{x}+\underline{\pmb{c}}^T\pmb{y}+\pmb{u}^T\pmb{b}\\
			\text{s.t.} & \pmb{H}(\pmb{y}+\pmb{x})\geq \pmb{g} \\
			& \pmb{x}+\pmb{y}\leq \pmb{1} \\
			& \pmb{u}^T\pmb{A}\geq \pmb{y}^T \\
			& \pmb{x}\in \{0,1\}^n \\
			& \pmb{y}, \pmb{u}\geq \pmb{0}
	\end{array}
\end{equation}

\begin{obs}
\label{polygap}
The integrality gap of~(\ref{mip2st}) is at least $\Omega(n)$ for the \textsc{RTSt Shortest Path} problem under the uncertainty set~$\cUhp_0$.
\end{obs}
\begin{proof}
Consider an instance of \textsc{RTSt Shortest Path} shown in Figure~\ref{fig0}. Set $\mathcal{X}$ contains characteristic vectors of  the simple $s-t$ paths from $s$ to $t$ of the form $s-i-t$, $i\in [m]$. Notice that $m=n/2$. It is easy to see that the optimal objective value of~(\ref{mip2st}) equals~$m$. 
In the relaxation of~(\ref{mip2st}) (see also the relaxation of~(\ref{mip2})) we can fix $x_{si}=\frac{1}{m}$, $y_{si}=0$ and $x_{it}=0$ and $y_{it}=\frac{1}{m}$ for each $i\in [m]$. The cost of this solution is~1, which gives the integrality gap of~$m=\Omega(n)$.
	\begin{figure}[ht]
		\centering
		\includegraphics[height=3.5cm]{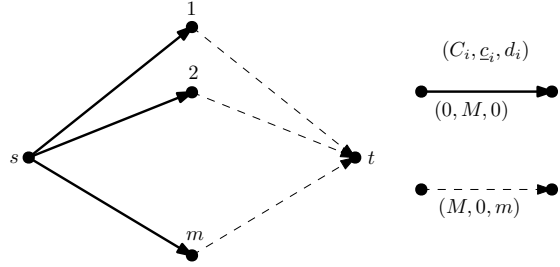}
		\caption{An instance of the robust two-stage shortest path problem with $\cUhp_0$,  $\Gamma=m$, and $M$ is a big constant.}\label{fig0}
	\end{figure}
\end{proof}

Problem~(\ref{mip2st}) can be solved in polynomial time for \textsc{RTSt Selection}  under $\cUhp_0$~\cite{CGKZ18}. In Section~\ref{secrs} we will show that the same result holds for \textsc{RTSt RS} under $\cUhp_0$. 
On the other hand,~(\ref{mip2st}) is strongly NP-hard for arbitrary $\cUhp$, when the constraint $\pmb{H}(\pmb{y}+\pmb{x})\geq \pmb{g}$ becomes $y_1+\dots+y_n+x_1+\dots+x_n= n$, i.e. when~(\ref{mip2st}) models the $\textsc{RTSt}_{\pmb{1}}$ problem (see Section~\ref{seccomplex}). We now show that $\textsc{RTSt}_{\pmb{1}}$ is polynomially solvable, when there is only a constant number of constraints in $\cUhp$, except for the nonnegativity constraints (note that the hardness result in Section~\ref{corhp} requires an unbounded number of constraints).

\begin{thm}
The $\textsc{RTSt}_{\pmb{1}}$ problem can be solved in polynomial time if the matrix $\pmb{A}$ in $\cUhp$ has a constant number of rows.
\end{thm}
\begin{proof}
Consider the formulation~\eqref{mip2st} for $\textsc{RTSt}_{\pmb{1}}$ with $\pmb{u}=(u_1,\ldots,u_m)$ for a constant $m$. Let us assume that $\pmb{x},\pmb{y}$ are fixed. The remaining optimization problem can be rewritten as the following linear program with additional $n$ slack variables $\pmb{s}$:
\begin{align*}
\min\ &\pmb{b}^T \pmb{u} \\
\text{s.t. } &\pmb{A}^T\pmb{u}- \pmb{s} = \pmb{y}  \\
& u_i \ge 0 & i\in[m] \\
& s_i \ge 0 &  i\in[n]
\end{align*}
where $\pmb{A}\in\Rset^{m\times n}$ and $\mathrm{rank}(\pmb{A})=m$.
The coefficient matrix of this problem is 
$\begin{bmatrix}\pmb{A}^T& {-\pmb{I}_n}\end{bmatrix}\in\mathbb{R}^{n\times (m+n)}$, where
 ${\pmb{I}_n}$ denotes the identity $n\times n$ matrix. Since $\cUhp$ is nonempty and bounded, there is an optimal $n \times n$ basis matrix $\pmb{B}$ to this problem, corresponding to basic variables $\pmb{u}_{\pmb{B}}$, $\pmb{s}_{\pmb{B}}$, so that
\begin{equation}
\begin{bmatrix}\pmb{u}_{\pmb{B}}\\ \pmb{s}_{\pmb{B}}\end{bmatrix} ={\pmb{B}^{-1}}\pmb{y}. 
\label{bos}
\end{equation}
We will use the fact that the matrix ${\pmb{B}^{-1}}$ has a special structure. Namely, by reordering the constraints and variables, 
we can assume that
\[{\pmb{B}^{-1}} = \begin{bmatrix}
\begin{array}{c|c}
\pmb{A}_1 & \begin{matrix}
{\pmb{\mathrm{O}}} 
\\
\hline 
{-\pmb{I}_{(n-m')}}
\end{matrix}
\end{array}
\end{bmatrix}^{-1} = 
\begin{bmatrix}
\begin{array}{c|c}
\pmb{A}_2 & \begin{matrix}
\pmb{\mathrm{O}}
\\
\hline 
{-\pmb{I}_{(n-m')}}
\end{matrix}
\end{array}
\end{bmatrix} \in \mathbb{R}^{n\times n}
\]
with $\pmb{A}_1, \pmb{A}_2 \in\mathbb{R}^{n\times m'}$  
{and  $\pmb{\mathrm{O}}$ being the $m'\times (n-m')$ zero matrix,
where $m'\leq m$ is} 
the size of
 $\pmb{u}_{\pmb{B}}$.
Fixing a basis matrix $\pmb{B}$, problem $\textsc{RTSt}_{\pmb{1}}$ thus simplifies to
\begin{align*}
\min\ &\pmb{C}^T \pmb{x} + \left( \underline{\pmb{c}}^T + 
\begin{bmatrix}{\pmb{b}^T_{\pmb{B}}} & {\pmb{0}^T_{\pmb{B}}}\end{bmatrix} \pmb{B}^{-1}\right) \pmb{y} \\
\text{s.t. } &\pmb{B}^{-1}\pmb{y}= \begin{bmatrix}
\begin{array}{c|c}
\pmb{A}_2 & \begin{matrix}
{\pmb{\mathrm{O}}} 
\\
\hline 
{-\pmb{I}_{(n-m')}}
\end{matrix}
\end{array}
\end{bmatrix} \pmb{y} \ge 0 \\
& \pmb{x}+\pmb{y}= \pmb{1} \\
& \pmb{x}\in \{0,1\}^n \\
& \pmb{y}\in\{0,1\}^n \\
\end{align*}
{where $\pmb{b}^T_{\pmb{B}}$  and $\pmb{0}^T_{\pmb{B}}$ 
are coefficients corresponding to 
 $\pmb{u}_{\pmb{B}}$ and $\pmb{s}_{\pmb{B}}$, respectively.}
Notice that $y_i\in \{0,1\}$ for each $i\in [n]$, because $x_i\in \{0,1\}$ and $x_i+y_i=1$ for all $i\in [n]$.
If we fix the values of the first $m'$ variables in $\pmb{y}$, corresponding to matrix $\pmb{A}_2$, the resulting problem can be solved in polynomial time. Indeed, in this case all the remaining variables in $\pmb{y}$ are either forced to~1, to~0, or are kept free.
There are $\binom{n+m}{n}=O( (m+n)^m )$ many different candidates to choose a basis, and for each candidate, we enumerate $O(2^m)$ values for the $\pmb{y}$-variables involved. For fixed $m$, the resulting complexity is thus polynomial in the input size.
\end{proof}

Let us now focus on ellipsoidal uncertainty.
If $\mathcal{U}=\cUe$, then~(\ref{mip1}) can be rewritten as
$$\min_{\{\pmb{y}:\pmb{H}(\pmb{y}+\pmb{x})\geq \pmb{g},\; \pmb{0}\leq \pmb{y}\leq \pmb{1}-\pmb{x}\}} \underline{\pmb{c}}^T\pmb{y}+ ||\pmb{A}^T\pmb{y}||_{2}.$$
Consequently, we get the following compact program for \textsc{RTSt} under $\cUe$:
\begin{equation}
\label{modelip}
	\begin{array}{lllll}
		\min & \pmb{C}^T\pmb{x}+\underline{\pmb{c}}^T\pmb{y}+||\pmb{A}^T\pmb{y}||_{2}\\
			\text{s.t.}  & \pmb{H}(\pmb{y}+\pmb{x})\geq \pmb{g} \\
			& \pmb{x}+\pmb{y}\leq \pmb{1} \\
			& \pmb{x}\in \{0,1\}^n \\
			& \pmb{y}\geq \pmb{0}
	\end{array}
\end{equation}
Problem~(\ref{modelip}) is a quadratic 0-1 optimization problem, which can be difficult to solve. In Section~\ref{apprgen} we will propose some methods of computing approximate solutions to~(\ref{modelip}).

\begin{obs}
The integrality gap of~(\ref{modelip}) is at least $\Omega(n)$ for the \textsc{RTSt Shortest Path} problem under the uncertainty set~$\cUe$.
\end{obs}
\begin{proof}
Consider the same network as in the proof of Observation~\ref{polygap}. For each arc $(s,i)$ we fix $C_{si}=0$ and $\underline{c}_{si}=M$ and for each arc $(i,t)$ we fix $C_{ti}=M$ an $\underline{c}_{ti}=0$, $i\in [m]$. Let $\pmb{A}$ be a $2m\times 2m$ diagonal matrix having the values of $m$ on the diagonal. Hence
$$\cUe=\{\underline{\pmb{c}}+m\pmb{\delta}: ||\pmb{\delta}||_{2}\leq 1\}\subset \Rset^{2m}.
$$
The reasoning is then the same as in the proof of Observation~\ref{polygap}.
\end{proof}

\section{Computing  approximate solutions}
\label{apprgen}


A compact formulation for the general \textsc{RTSt} problem is unknown.  Therefore, solving the problem requires applying special row and column generation techniques (see, e.g.~\cite{ZZ13}).  As this method may consist of solving many hard MIP formulations, it can be inefficient for large problems. In this section we propose algorithms, which return solutions with some guaranteed distance to the optimum. We will discuss a general case as well as  cases that can be modeled as the min-max problem~(\ref{mip2}).
  
\subsection{General approximation results} 
 
 Let  $\mathcal{X}$ be expressed as~(\ref{eqx}), but now no assumptions on
 the polyhedron $\mathcal{N}$ (see~(\ref{polh})) are imposed. So, the underlying deterministic problem can be NP-hard and also hard to approximate. By interchanging the min-max operators we get the following lower bound on the optimal objective value
of the \textsc{RTSt} problem:
$$LB=\max_{\pmb{c}\in \mathcal{U}} \min_{\pmb{x}\in \mathcal{X}'}\min_{\pmb{y}\in \mathcal{R}(\pmb{x})} (\pmb{C}^T\pmb{x}+\pmb{c}^T\pmb{y})=\max_{\pmb{c}\in \mathcal{U}} \min_{(\pmb{x},\pmb{y}) \in \mathcal{Z}} (\pmb{C}^T\pmb{x}+\pmb{c}^T\pmb{y}),$$
where
 $$
		\mathcal{Z}=\{(\pmb{x},\pmb{y}): \pmb{H}(\pmb{x}+\pmb{y})\geq \pmb{g}, \pmb{x}+\pmb{y}\leq \pmb{1}, \pmb{x} \in\{0,1\}^n,\pmb{y}\in\{0,1\}^n\}.
 $$
 Consider the following relaxation of $\mathcal{Z}$:
  $$
		\mathcal{Z}'=\{(\pmb{x},\pmb{y}): \pmb{H}(\pmb{x}+\pmb{y})\geq \pmb{g}, \pmb{x}+\pmb{y}\leq \pmb{1}, \pmb{0}\leq \pmb{x}\leq \pmb{1}, \pmb{0} \leq \pmb{y}\leq \pmb{1}\}.
 $$
 Since $\mathcal{U}$ and $\mathcal{Z}'$ are convex sets,  by the minimax theorem~\cite{N28},  we have
 $$LB\geq \max_{\pmb{c}\in \mathcal{U}} \min_{(\pmb{x},\pmb{y})\in \mathcal{Z}'} (\pmb{C}^T\pmb{x}+\pmb{c}^T\pmb{y})=
 \min_{(\pmb{x},\pmb{y})\in \mathcal{Z}'}\max_{\pmb{c}\in \mathcal{U}}  (\pmb{C}^T\pmb{x}+\pmb{c}^T\pmb{y})$$
We also get the following upper bound on the optimal objective value (the \emph{min-max problem}): 
 \begin{equation}
 \label{minmax}
 UB=\min_{\pmb{x}\in \mathcal{X}'}\min_{\pmb{y}\in \mathcal{R}(\pmb{x})}\max_{\pmb{c}\in \mathcal{U}}  (\pmb{C}^T\pmb{x}+\pmb{c}^T\pmb{y})=\min_{(\pmb{x},\pmb{y})\in \mathcal{Z}}\max_{\pmb{c}\in \mathcal{U}}  (\pmb{C}^T\pmb{x}+\pmb{c}^T\pmb{y}).
 \end{equation}
 We thus get
 \begin{equation}
 \label{eqfract}
 \frac{UB}{LB}\leq \frac{\displaystyle \min_{(\pmb{x},\pmb{y})\in \mathcal{Z}}\max_{\pmb{c}\in \mathcal{U}}  (\pmb{C}^T\pmb{x}+\pmb{c}^T\pmb{y})}{\displaystyle \min_{(\pmb{x},\pmb{y})\in \mathcal{Z}'}\max_{\pmb{c}\in \mathcal{U}}  (\pmb{C}^T\pmb{x}+\pmb{c}^T\pmb{y})}=\rho.
 \end{equation}
 Let $(\pmb{x}^*,\pmb{y}^*)\in \mathcal{Z}$ be an optimal solution to the min-max problem~(\ref{minmax}). Then 
 $$\textsc{Eval}(\pmb{x}^*)\leq \pmb{C}^T\pmb{x}^*+\max_{\pmb{c}\in \mathcal{U}} \pmb{c}^T \pmb{y}^*=UB.$$
 We thus get 
 \begin{equation}
 \label{eqeval}
 \textsc{Eval}(\pmb{x}^*)\leq \rho \cdot LB
 \end{equation}
 and $\pmb{x}^*\in \mathcal{X}'$ is a $\rho$-\emph{approximate, first-stage solution} to~\textsc{RTSt},
  i.e. a solution whose value $\textsc{Eval}(\pmb{x}^*)$ is within a
factor of $\rho$ of the value of an optimal solution to~\textsc{RTSt}.
  For the uncertainty sets $\cUhp$ and $\cUe$ the value of LB can be computed in polynomial time by solving convex optimization problems and
 for $\cUvp$ by solving an LP problem. On the other hand, the upper bound and approximate solution~$\pmb{x}^*$ can be computed by solving a compact 0-1 problem (after dualizing the inner maximization problem in~(\ref{minmax})). In the next part of this section we will show a special case of the problem for which $\pmb{x}^*$ can be computed in polynomial time.

We now consider the polyhedral uncertainty. Using duality, the min-max problem~(\ref{minmax}) under $\cUhp$, can be represented as the following MIP formulation:
\begin{equation}
\label{mipminmax}
	\begin{array}{lllll}
		\min & \pmb{C}^T\pmb{x}+\underline{\pmb{c}}^T\pmb{y}+\pmb{u}^T\pmb{b}\\
			\text{s.t.}  & \pmb{H}(\pmb{y}+\pmb{x})\geq \pmb{g} \\
			& \pmb{x}+\pmb{y}\leq \pmb{1} \\
			& \pmb{u}^T\pmb{A}\geq \pmb{y}^T \\
			& \pmb{x}, \pmb{y}\in \{0,1\}^n \\
			& \pmb{u}\geq \pmb{0}
	\end{array}
\end{equation}
 The relaxation of~(\ref{mipminmax}), used to compute $LB$, is an LP problem, so it can be solved in polynomial time. The problem~(\ref{mipminmax}) can be more complex. However, it can be easier to solve than the original robust two-stage problem. Using~(\ref{eqfract}) and~(\ref{eqeval}), we get the following theorem:
 \begin{thm}
 	Let $\pmb{x}^*$ be optimal to~(\ref{mipminmax}). Then $\pmb{x}^*$ is a $\rho$-approximate first-stage solution to the \textsc{RTSt} problem and $\rho$ is the integrality gap of~(\ref{mipminmax}).
 \end{thm}
 
 We now describe the case in which $\pmb{x}^*$ can be computed in polynomial time, which yields a $\rho$-approximation algorithm for the robust two-stage problem. Namely, we consider the continuous budgeted uncertainty $\cUhp_0$. Fix $\pmb{y}$ and consider the following problem:
 $$\max_{\pmb{c}\in \cUhp_0} \pmb{c}^T\pmb{y}.$$
This problem can be solved by observing that either the whole budget $\Gamma$ is allocated to $\pmb{y}$ or the allocation is blocked by the upper bounds on the deviations. So
$$\max_{\pmb{c}\in \cUhp_0} \pmb{c}^T\pmb{y}=\min\{\underline{\pmb{c}}^T\pmb{y}+\Gamma, (\underline{\pmb{c}}+\pmb{d})^T\pmb{y}\}.$$
Hence the min-max problem can be rewritten as follows:
\begin{align*}
&\min_{(\pmb{x},\pmb{y})\in \mathcal{Z}}\max_{\pmb{c}\in \cUhp_0}  (\pmb{C}^T\pmb{x}+\pmb{c}^T\pmb{y}) \\
=&\min_{(\pmb{x},\pmb{y})\in \mathcal{Z}}\min\{\pmb{C}^T\pmb{x}+\underline{\pmb{c}}^T\pmb{y}+\Gamma, \pmb{C}^T\pmb{x}+(\underline{\pmb{c}}+\pmb{d})^T\pmb{y}\} \\
=&\min\{\textsc{TSt}(\underline{\pmb{c}})+\Gamma, \textsc{TSt}(\underline{\pmb{c}}+\pmb{d})\}.
\end{align*}
In consequence, the minmax problem reduces to solving two two-stage problems, which can be done in polynomial time if the underlying problem $\mathcal{P}$ is polynomially solvable (see Observation~\ref{obsones}). So, in this case a $\rho$-approximate solution $\pmb{x}^*$ can be computed in polynomial time.

\subsection{Approximating the problems with the integrality property}

In this section we propose some methods of constructing approximate solutions for 
the \textsc{RTSt} problem if the polyhedron $\mathcal{N}$ (see~(\ref{polh}))
satisfies the integrality property. Recall that in this case we can represent  \textsc{RTSt} as the 
min-max formulation~(\ref{mip2}), so from now on we  explore the approximability of~(\ref{mip2}).
 Let $\tilde{\pmb{c}}\in\cU$ be any fixed scenario.  Thus the two-stage problem (see Section~\ref{secform})
 with~$\tilde{\pmb{c}}$, in the second stage,
 can be then formulated as follows:
\begin{equation}
\label{miptst}
	\begin{array}{llll}
	\min & \pmb{C}^T\pmb{x}+\tilde{\pmb{c}}^T\pmb{y} \\
		& \pmb{H}(\pmb{x}+\pmb{y})\geq \pmb{g}\\
		& \pmb{x}+\pmb{y}\leq \pmb{1}\\
		& \pmb{x}, \pmb{y}\in \{0,1\}^n
	\end{array}
\end{equation}
Using Observation~\ref{obsones}, we can solve~(\ref{miptst}) in polynomial time, by solving 
 one underlying deterministic problem~$\mathcal{P}$. We now show how to obtain an approximate solution to~(\ref{mip2}) by solving~(\ref{miptst}) for an appropriately chosen scenario $\tilde{\pmb{c}}$. Let $(\hat{\pmb{x}}, \hat{\pmb{y}})\in \mathcal{Z}$ be an optimal solution to~(\ref{miptst}). 
\begin{lem}\label{applemma}
If $c_i \le t \tilde{c}_i$, $i\in [n]$ (shortly $\pmb{c} \le t \tilde{\pmb{c}}$) for each $\pmb{c}\in \cU$, then $(\hat{\pmb{x}},\hat{\pmb{y}})$ is a $t$-approximate solution to~(\ref{mip2}).
\end{lem}
\begin{proof}
Let $(\pmb{x}^*,\pmb{y}^*)$ be an optimal solution to~(\ref{mip2}).  We then have
$\max_{\pmb{c}\in\cU} (\pmb{C}^T\hat{\pmb{x}} + \pmb{c}^T\hat{\pmb{y}})= \pmb{C}^T\hat{\pmb{x}} +
 \overline{\pmb{c}}^T\hat{\pmb{y}}
 \overset{(1)}{\le} t \left( \pmb{C}^T\hat{\pmb{x}} + \tilde{\pmb{c}}^T \hat{\pmb{y}} \right) \overset{(2)}{\le} t \left( \pmb{C}^T\pmb{x}^* + \tilde{\pmb{c}}^T\pmb{y}^* \right) \le t \left(  \pmb{C}^T\pmb{x}^* + \max_{\pmb{c}\in\cU} \pmb{c}^T\pmb{y}^* \right)$. The inequality~(1) follows from the assumption that $\pmb{c}^*\leq t \tilde{\pmb{c}}$ and $t\geq 1$. The inequality~(2) holds
 because  $(\hat{\pmb{x}}, \hat{\pmb{y}})$ is an optimal solution to~(\ref{miptst}) and this optimal solution will not change when we relax $\pmb{y}\in \{0,1\}^n$ with $\pmb{0}\leq \pmb{y}\leq \pmb{1}$ in~(\ref{miptst}) due to 
  the integrality property assumed.
\end{proof}
Accordingly, we can construct the best guarantee $t$,  by solving the following convex optimization problem:
\begin{equation*}
\begin{array}{lllll}
\max\ &t^{-1} \\
\text{s.t. } & t^{-1} \max_{\pmb{c}\in\cU} c_i \le \tilde{c}_i & i\in[n] \\
& \tilde{\pmb{c}}\in\cU \\
& t^{-1} \ge 0
\end{array}
\end{equation*}
where the values $\max_{\pmb{c}\in\cU} c_i$, $i\in [n]$, have to be precomputed by solving additional $n$ convex problems. 

\subsubsection{Polyhedral uncertainty}

The next two theorems are consequences of Lemma~\ref{applemma}.
\begin{thm}
\label{apprK}
Problem~(\ref{mip2}) with $\cUvp={\rm conv}\{\pmb{c}_1,\dots,\pmb{c}_K\}$ is approximable within $K$.
\end{thm}
\begin{proof}
Fix $\tilde{\pmb{c}}=\frac{1}{K} \sum_{k\in[K]} \pmb{c}_k\in \cUvp$. Then for each $\pmb{c}\in {\rm conv}\{\pmb{c}_1,\dots,\pmb{c}_K\}$,  the inequality $\pmb{c}\leq K\tilde{\pmb{c}}$ holds. Thus by fixing $t=K$ in 
 Lemma~\ref{applemma} the theorem follows.
\end{proof}
\begin{thm}
	If  $\underline{\pmb{c}} \geq \alpha (\underline{\pmb{c}}+\pmb{d})$, $\alpha \in (0,1]$, in $\cUhp_0$, then~(\ref{mip2}) with $\cUhp_0$ is approximable within~$\frac{1}{\alpha}$.
\end{thm}
\begin{proof}
Fix $\tilde{\pmb{c}}=\underline{\pmb{c}}$. Then for each scenario $\pmb{c}\in \cUhp_0$, we get $\pmb{c}\leq \underline{\pmb{c}}+\pmb{d}\leq \frac{1}{\alpha}\tilde{\pmb{c}}$. Thus by fixing $t=\frac{1}{\alpha}\geq 1$ in 
 Lemma~\ref{applemma} the result follows.
\end{proof}

 The next result characterizes the approximability of the problem under  $\cUhp_1$.
\begin{thm}
\label{thmfptasgen}
	Assume that the number of budget constraints in $\cUhp_1$ is constant and the following problem is polynomially solvable:
	\begin{equation}
\label{mip2tstp}
	\begin{array}{lllll}
		\min & \pmb{C}^T\pmb{x}+\underline{\pmb{c}}^T\pmb{y}\\
			\text{\rm s.t.} & \pmb{H}(\pmb{y}+\pmb{x})\geq \pmb{g} \\
			& \pmb{x}+\pmb{y}\leq \pmb{1} \\
			& \pmb{0}\leq \pmb{y} \leq \pmb{d}\\
			& \pmb{x}\in \{0,1\}^n \\
	\end{array}
\end{equation}
where $d_i\in \mathcal{E}=\{0,\epsilon, 2\epsilon, \dots, 1\}$, $i\in [n]$. Then~(\ref{mip2}) under $\cUhp_1$ admits an FPTAS.
\end{thm}
\begin{proof}
	The compact MIP formulation~(\ref{mip2st})  for~(\ref{mip2}) takes the following form:
	\begin{equation}
	\label{mipH1}
		 \begin{array}{lllll}
		\min & \pmb{C}^T\pmb{x}+\underline{\pmb{c}}^T\pmb{y}+ \displaystyle\sum_{j\in [K]} u_j\Gamma_j\\
			\text{s.t.} & \pmb{H}(\pmb{y}+\pmb{x})\geq \pmb{g} \\
			& \pmb{x}+\pmb{y}\leq \pmb{1}\\
			& \displaystyle \sum_{\{j\in [K]: i\in U_j\}} u_j\geq y_i & i\in [n] \\
			& \pmb{x}\in \{0,1\}^n \\
			& \pmb{y}, \pmb{u}\geq \pmb{0}
	\end{array}
	\end{equation}
	Since $y_i\in [0,1]$ for each $i\in [n]$, we get $u_j\in [0,1]$ for each $j\in [K]$. Let us fix $\epsilon=\frac{1}{t}$ for some integer $t\geq 0$, and consider the numbers $\mathcal{E}=\{0,\epsilon, 2\epsilon, 3\epsilon, \dots,1\}$.  Fix vector $(u_1,\dots,u_K)$, where $u_j\in \mathcal{E}$. The problem~(\ref{mipH1}) reduces then to~(\ref{mip2tstp}),
	where $d_i=\min\{\sum_{\{j\in [K]: i\in U_j\}} u_j,1\}$, $i\in [n]$. Let us enumerate all $(\frac{1}{\epsilon})^K$ vectors $\pmb{u}$, with components $u_j\in \mathcal{E}$, $j\in [K]$, and let us solve~(\ref{mip2tstp}) for each such a vector. Assume that $(\hat{\pmb{x}},\hat{\pmb{y}}, \hat{\pmb{u}})$ is the enumerated solution having the minimum objective value in~(\ref{mipH1}) (notice that $(\hat{\pmb{x}},\hat{\pmb{y}}, \hat{\pmb{u}})$ is feasible to~(\ref{mipH1})). Let $(\pmb{x}^*,\pmb{y}^*, \pmb{u}^*)$ be an optimal solution to~(\ref{mipH1}). Let us round up the components of $\pmb{u}^*$ to the nearest values in $\mathcal{E}$. As the result we get a feasible solution with the cost at most $(1+\epsilon)$ greater than the optimum.  Furthermore the cost of this solution is not greater than the cost of $(\hat{\pmb{x}},\hat{\pmb{y}}, \hat{\pmb{u}})$, because the rounded vector $\pmb{u}^*$ has been enumerated. By the assumption that $K$ is constant and~(\ref{mipH1}) can be solved in polynomial time, we get an FPTAS for~(\ref{mip2}) under $\cUhp_1$.
	\end{proof}

 We will show how to solve~(\ref{mip2tstp}) for particular problems in Sections~\ref{secsel},~\ref{secrs} and~\ref{secsp}.

\subsubsection{Ellipsoidal uncertainty}

In this section we will focus on constructing approximate solutions to~(\ref{modelip}), which is a compact formulation of~(\ref{mip2}) under ellipsoidal uncertainty $\cUe$. As~(\ref{modelip}) is a 0-1 quadratic problem, it can be hard to solve.
Consider the following linearization of~(\ref{modelip}):
\begin{equation}
\label{modelip1}
	\begin{array}{lllll}
		\min & \pmb{C}^T\pmb{x}+\underline{\pmb{c}}^T\pmb{y}+||\pmb{A}^T\pmb{y}||_{1}\\
			\text{s.t.}  & \pmb{H}(\pmb{y}+\pmb{x})\geq \pmb{g} \\
			& \pmb{x}+\pmb{y}\leq \pmb{1} \\
			& \pmb{x}\in \{0,1\}^n \\
			& \pmb{y}\geq \pmb{0}
	\end{array}
\end{equation}
which can be represented as the following linear MIP problem:
\begin{equation}
\label{modelip2}
	\begin{array}{lllll}
		\min & \pmb{C}^T\pmb{x}+\underline{\pmb{c}}^T\pmb{y}+ \displaystyle\sum_{i\in [n]} z_i\\
			\text{s.t.}  & \pmb{H}(\pmb{y}+\pmb{x})\geq \pmb{g} \\
			& \pmb{x}+\pmb{y}\leq \pmb{1} \\
			& z_i\geq \pmb{A}_i^T\pmb{y} & i\in [n] \\
			& z_i\geq -\pmb{A}_i^T\pmb{y} & i\in [n]\\
			& \pmb{x}\in \{0,1\}^n \\
			& \pmb{y}\geq \pmb{0}
	\end{array}
\end{equation}
where $\pmb{A}_i$ is the $i$th column of $\pmb{A}$. 
\begin{thm}
\label{thmappreps}
	Let $(\hat{\pmb{x}},\hat{\pmb{y}})$ be an optimal solution to~(\ref{modelip2}) and $(\pmb{x}^*,\pmb{y}^*)$ be an optimal solution to~(\ref{modelip}). Then $\textsc{Eval}(\hat{\pmb{x}})\leq \sqrt{n}\cdot \textsc{Eval}(\pmb{x}^*)$
\end{thm}
\begin{proof}
	We use the following well known inequalities:
	$$\frac{1}{\sqrt{n}}\cdot ||\pmb{A}^T\pmb{y}||_{1}\leq ||\pmb{A}^T\pmb{y}||_{2} \leq ||\pmb{A}^T\pmb{y}||_{1}$$
	Using them, we get
	$$\pmb{C}^T\hat{\pmb{x}}+\underline{\pmb{c}}^T\hat{\pmb{y}}+||\pmb{A}^T\hat{\pmb{y}}||_{2}\leq \pmb{C}^T\hat{\pmb{x}}+\underline{\pmb{c}}^T\hat{\pmb{y}}+||\pmb{A}^T\hat{\pmb{y}}||_{1} \leq \pmb{C}^T\pmb{x}^*+\underline{\pmb{c}}^T\pmb{y}^*+||\pmb{A}^T\pmb{y}^*||_{1}$$
$$
\leq \pmb{C}^T\pmb{x}^*+\underline{\pmb{c}}^T\pmb{y}^*+\sqrt{n}\cdot ||\pmb{A}^T\pmb{y}^*||_{2}
\leq\sqrt{n}\cdot (\pmb{C}^T\pmb{x}^*+\underline{\pmb{c}}^T\pmb{y}^*+||\pmb{A}^T\pmb{y}^*||_{2}),
$$
and the theorem follows.
\end{proof}
Problem~(\ref{modelip2}) is a linear MIP, so it can be easier to solve than~(\ref{modelip}). Unfortunately, it is still NP-hard even if the underlying deterministic problem is polynomially solvable.
\begin{obs}
	Problem~(\ref{modelip2}) is NP-hard when $\mathcal{X}_{\pmb{1}}=\{\pmb{x}\in\{0,1\}^n: x_1+\dots+x_n=n\}$.
\end{obs}
\begin{proof}
	It follows directly from the proof of Theorem~\ref{thmcomplE}. It is easy to see that $||\pmb{A}^T\pmb{y}||_{2}=||\pmb{A}^T\pmb{y}||_{1}$ for the matrix $\pmb{A}$ constructed in the proof.
\end{proof}

\begin{thm}
	If all the entries of $\pmb{A}$ are nonnegative and $\mathcal{P}$ is polynomially solvable, then~(\ref{modelip}) is approximable within $\sqrt{n}$.
\end{thm}
\begin{proof}
	If all the entries of $\pmb{A}=[a_{ij}]$ are nonnegative, then ~(\ref{modelip1}) can be rewritten as follows:
	\begin{equation}
\label{modelip3}
	\begin{array}{lllll}
		\min & \pmb{C}^T\pmb{x}+\underline{\pmb{c}}^T\pmb{y}+\displaystyle\sum_{i\in [n]} \pmb{A}^T_i\pmb{y}\\
			\text{s.t.}  & \pmb{H}(\pmb{y}+\pmb{x})\geq \pmb{g} \\
			& \pmb{x}+\pmb{y}\leq \pmb{1} \\
			& \pmb{x}\in \{0,1\}^n \\
			& \pmb{y}\geq \pmb{0}
	\end{array}
\end{equation}
which is equivalent to
\begin{equation}
\label{modelip4}
	\begin{array}{lllll}
		\min & \displaystyle\sum_{i\in [n]} (C_i x_i + \hat{c}_i y_i)\\
			\text{s.t.}  & \pmb{H}(\pmb{y}+\pmb{x})\geq \pmb{g} \\
			& \pmb{x}+\pmb{y}\leq \pmb{1} \\
			& \pmb{x}\in \{0,1\}^n \\
			& \pmb{y}\geq \pmb{0}
	\end{array}
\end{equation}
where $\hat{c}_i=\underline{c}_i+\sum_{j\in [n]} a_{ji}$, $i\in [n]$. Problem~(\ref{modelip4}) is a two-stage problem with one second stage scenario~$\hat{\pmb{c}}$ and it is polynomially solvable according to Observation~\ref{obsones}, if problem~$\mathcal{P}$ is solvable in polynomial time. Notice that relaxing $\pmb{y}\in \{0,1\}^n$ with $\pmb{y}\geq \pmb{0}$ does not change an optimal solution to~(\ref{modelip4}), due to the integrality assumption.
Now Theorem~\ref{thmappreps} implies the result.
\end{proof}

\section{Robust two-stage selection problem}
\label{secsel}

In this section we investigate in more detail the robust two-stage version of the \textsc{Selection} problem under $\cUhp$ and $\cUe$. In Section~\ref{seccomplex} we have proved that this problem is NP-hard. Let us also recall that \textsc{RTSt Selection} is polynomially solvable under $\cUhp_0$~\cite{CGKZ18}. The MIP formulations~(\ref{mip2st}) and~(\ref{modelip}) for the problem under $\cUhp$ and $\cUe$, respectively,   take the following form
\begin{equation}
\label{mip2stsel0}
\begin{array}{lll}
	\begin{array}{llllll}
		(a) &\min & \pmb{C}^T\pmb{x}+\underline{\pmb{c}}^T\pmb{y}+\pmb{u}^T\pmb{b}\\
		&	\text{s.t.}  & \displaystyle \sum_{i\in [n]} (x_i+y_i) = p \\
		&	& \pmb{x}+\pmb{y}\leq \pmb{1} \\
		&	& \pmb{u}^T\pmb{A}\geq \pmb{y}^T \\
		&	& \pmb{x}\in \{0,1\}^n \\
		&	& \pmb{y}, \pmb{u}\geq \pmb{0}
	\end{array}
	&
	 \begin{array}{llllllll}
	(b)&\min\ &\pmb{C}^T\pmb{x} + \underline{\pmb{c}}^T\pmb{y} + \Vert \pmb{A}^T\pmb{y}\Vert_2 \\
 &\text{s.t.}  & \displaystyle \sum_{i\in[n]} (x_i + y_i) = p \\
&& \pmb{x} + \pmb{y} \le  \pmb{1} \\
&& \pmb{x}\in\{0,1\}^n \\
&& \pmb{y} \ge 0\\
&&
\end{array}
\end{array}
\end{equation}
We first show the following approximation result:
\begin{thm}\label{select-2-appr}
	The \textsc{RTSt Selection} problem with uncertainty $\cUhp$ is approximable within~2.
\end{thm}
\begin{proof}
	Assume w.l.o.g that $C_1\leq C_2\leq \dots \leq C_n$.
	Consider the following LP relaxation of~(\ref{mip2stsel0})a:
	\begin{equation}
\label{mip2stselrel}
	\begin{array}{rlllll}
		LB=&\min & \pmb{C}^T\pmb{x}+\underline{\pmb{c}}^T\pmb{y}+\pmb{u}^T\pmb{b}\\
			&\text{s.t.}  & \displaystyle\sum_{i\in [n]} (x_i+y_i) = p \\
			&& \pmb{x}+\pmb{y}\leq \pmb{1} \\
			&& \pmb{u}^T\pmb{A}\geq \pmb{y}^T \\
			&& \pmb{0}\leq \pmb{x} \leq \pmb{1} \\
			&& \pmb{y}, \pmb{u}\geq \pmb{0}
	\end{array}
\end{equation}
Let $(\pmb{x}^*, \pmb{y}^*, \pmb{u}^*)$ be an optimal solution to~(\ref{mip2stselrel}). We first note that given $\pmb{y}^*$, the optimal values of $\pmb{x}^*$ can be obtained in the following greedy way. Set $p^*:=p-\sum_{i\in [n]} y^*_i$. For $i:=1,\dots,n$, assign $x^*_i:=\min\{p^*, 1-y^*_i\}$ and update $p^*:=p^*-x^*_i$. Let $\ell\in [n]$ be such that $x^*_i>0$ for every $i\leq \ell$ and $x^*_i=0$ for every $i>\ell$. It is easily seen that $x^*_i+y^*_i=1$ for all $i\in [\ell-1]$. Therefore the quantity $p-\sum_{i\in [\ell-1]} (x^*_i+y^*_i)$ must be integral. By the construction, we get
$$\sum_{i\in [n]} (x^*_i+y^*_i)= p.$$
Notice also that $0<x^*_{\ell}+y^*_{\ell}<1$ may happen.

We now construct a feasible solution $(\hat{\pmb{x}}, \hat{\pmb{y}}, \hat{\pmb{u}})$ to~(\ref{mip2stsel0}) in the following way.
Set $\hat{p}:=p$.
 For $i:=1,\dots,\ell-1$,  if $x_i^*\geq \frac{1}{2}$, then assign $\hat{x}_i:=1$ and $\hat{y}_i:=0$; otherwise ($y^*_i\geq \frac{1}{2}$)
 assign  $\hat{x}_i:=0$ and $\hat{y}_i:=1$; and update  $\hat{p}:=\hat{p}-\ell+1$.
If $x^*_\ell\geq \frac{1}{2}$, then assign $\hat{x}_\ell:=1$ and $\hat{y}_\ell:=0$ and update  $\hat{p}:=\hat{p}-\hat{x}_\ell$;
otherwise ($x^*_\ell< \frac{1}{2}$)
assign $\hat{x}_\ell:=0$ and $\hat{y}_\ell:=\min\{1,2 y^*_\ell\}$  and update  $\hat{p}:=\hat{p}-\hat{y}_\ell$.
For $i:=\ell+1,\dots,n$, 
assign $\hat{x}_i:=0$ and $\hat{y}_i:=\min\{1,2 y^*_i,\hat{p}\}$  and update  $\hat{p}:=\hat{p}-\hat{y}_i$.
Finally assign $\hat{\pmb{u}}:=2\pmb{u}^*$.

We now need to show that  $(\hat{\pmb{x}}, \hat{\pmb{y}}, \hat{\pmb{u}})$ is a feasible solution to~(\ref{mip2stsel0})a.
It is clear that $\hat{\pmb{x}}\in\{0,1\}^n$ and 
$\hat{\pmb{x}}+\hat{\pmb{y}} \leq \pmb{1}$.
The constraints  $\hat{\pmb{u}}^T\pmb{A}\geq \hat{\pmb{y}}^T$ are satisfied, because
$\pmb{u}^{*T}\pmb{A}_i\geq y^*_i$, which yields $2\pmb{u}^{*T}\pmb{A}_i\geq 2y^*_i\geq \hat{y}_i$ for each $i\in [n]$,
where $\pmb{A}_i$ is the $i$th column of $\pmb{A}$. 
It remains to prove that $\sum_{i\in [n]} (\hat{x}_i+\hat{y}_i)= p$, i.e. $\hat{p}=0$ after 
the termination of the above algorithm.
We see  at once that $\sum_{i\in [\ell-1]} (\hat{x}_i+\hat{y}_i)=\ell-1$, since $\sum_{i\in [\ell-1]} (x^*_i+y^*_i)=\ell-1$.

We now show that $\hat{x}_{\ell}+\hat{y}_{\ell}+\sum_{i> \ell} \hat{y}_i=p-(\ell-1)$.
 After assigning  the first  $\ell-1$variables, $\hat{p}$ satisfies 
\begin{equation}
\hat{p}=x^*_{\ell}+y^*_{\ell}+\sum_{i> \ell}y^*_i=p-(\ell-1).
\label{ephat}
\end{equation}
We need to consider only two cases.
The first one:  $x^*_{\ell}\leq y^*_{\ell}$ or $\frac{1}{2}\leq x^*_{\ell}$.
For  $x^*_{\ell}\leq y^*_{\ell}$, $\hat{x}_{\ell}=0$ and
$\hat{y}_{i}\leq \min\{1,2y^*_i\}$ for each $i\geq \ell$.
According to~(\ref{ephat}), we have $\hat{p}\leq \sum_{i\geq \ell}\min\{1,2y^*_i\}$.
Hence one can allocate feasible values to $\hat{y}_{i}$, $i\geq \ell$,
until $\sum_{i\geq \ell}\hat{y}_i$ reaches~$\hat{p}$.
In the case: $\frac{1}{2}\leq x^*_{\ell}$, $\hat{y}_{\ell}=0$,
$\hat{x}_{\ell}=1$ and 
$\hat{y}_{i}\leq \min\{1,2y^*_i\}$ for each $i> \ell$.
By~(\ref{ephat}), we get $\hat{p}\leq 1+\sum_{i> \ell}\min\{1,2y^*_i\}$.
Again one can pack $\hat{x}_{\ell}=1$ and $\hat{y}_{i}$, $i> \ell$, until
$\hat{x}_{\ell}+\sum_{i> \ell}\hat{y}_i$  reaches~$\hat{p}$.

The second case: $y^*_{\ell}<x^*_{\ell}<\frac{1}{2}$. 
 We show that $\hat{y}_{\ell}+\sum_{i> \ell} \hat{y}_i=\hat{p}$ only for worst case value distributions of variables
 $x^*_{\ell}$ and  $y^*_{i}$, $i\geq \ell$, i.e.
for  distributions, where the values are as follows:
 $y^*_i=1$ for every $i=\ell+1,\ldots,\ell+\hat{p}-1$,  and $y^*_{\ell+\hat{p}}=1-(x^*_{\ell}+y^*_{\ell})$
 (a similar  reasoning applies to other  distributions).
 Thus $\hat{x}_{\ell}=0$, $\hat{y}_{\ell}\leq \min\{1,2y^*_\ell\}=2y^*_\ell$,
$\hat{y}_{i}\leq \min\{1,2y^*_i\}=y^*_i$ for each  $i=\ell+1,\ldots,\ell+\hat{p}-1$, and
 $\hat{y}_{\ell+\hat{p}}\leq \min\{1,2(1-(x^*_{\ell}+y^*_{\ell}))\}$.
 From~(\ref{ephat}) and the assumption $x^*_{\ell}<\frac{1}{2}$, we obtain
 $\hat{p}\leq 2y^*_\ell+\sum_{i=\ell+1}^{\ell+\hat{p}-1}y^*_i+ \min\{1,2(1-(x^*_{\ell}+y^*_{\ell}))\}$.
 In consequence  one can allocate values to $\hat{y}_i$, $i\geq\ell$, to satisfy 
 $\hat{y}_{\ell}+\sum_{i> \ell} \hat{y}_i=\hat{p}$.

 The total cost of the feasible solution $(\hat{\pmb{x}}, \hat{\pmb{y}}, \hat{\pmb{u}})$ is at most 
twice the optimal value. Indeed,
$$\pmb{C}^T\hat{\pmb{x}}+\underline{\pmb{c}}^T\hat{\pmb{y}}+\hat{\pmb{u}}^T\pmb{b}\leq 2\cdot (\pmb{C}^T\pmb{x}^*+\underline{\pmb{c}}^T\pmb{y}^*+\pmb{u}^{*T}\pmb{b})=2\cdot LB,$$
and the proof is complete.
\end{proof}

\begin{thm}
\label{apprguar}
The approximation guarantee of the rounding algorithm presented in the proof of Theorem~\ref{select-2-appr} is tight, even if $p=n$ and $\cUhp$ has a single constraint.
\end{thm}
\begin{proof}
We consider the following problem instance: $p=n=2$, 
$\pmb{C}=\begin{bmatrix}10\\ \gamma \end{bmatrix}$, 
\[\cU=\left\{ \begin{bmatrix}0\\ \epsilon\end{bmatrix} + \begin{bmatrix}\delta_1\\ \delta_2\end{bmatrix} : \delta_1 +
\left (\frac{1}{2}+\mu \right)\delta_2 \le 1 \right\}\]
with $\mu>0$ and $\gamma > \epsilon > 0$ being small values. Then the compact formulation is
\begin{align*}
\min\ &10x_1 + \gamma x_2 + \epsilon y_2 + u \\
\text{s.t. }  & x_1 + y_1 + x_2 + y_2 =2 \\
& x_1 + y_1 \le 1 \\
& x_2 + y_2 \le 1 \\
&u \ge y_1 \\
&\left(\frac{1}{2}+\mu \right)u \ge y_2 \\
& x_1,x_2 \in\{0,1\} \\
& y_1,y_2 \in [0,1] \\
& u \ge 0
\end{align*}
An optimal solution to this problem is to set $y_1=x_2=u=1$, with objective function $1+\gamma$. An optimal solution for the LP relaxation of this problem is $y_1 = u = 1$, $x_2=\frac{1}{2}-\mu$ and $y_2=\frac{1}{2}+\mu$.  Applying our algorithm, we round $y_2$ to 1, which means that $u$ has to be increased to $2(1+\mu)$. The objective value of this solution is $2(1+\mu)+\epsilon$. As $\mu,\gamma,\epsilon$ approach 0, the ratio of optimal objective value and objective value of the approximate solution approaches 2.
\end{proof}

\begin{thm}
\label{thmgapsel}
The integrality gap of problem~(\ref{mip2stsel0})a is at least 4/3.
\end{thm}
\begin{proof}
Consider the problem with $n=p=2$, 
$\pmb{C}=\begin{bmatrix}10\\ 1 \end{bmatrix}$, 
\[\cU=\left\{ \pmb{0} + \pmb{\delta} : \delta_1 + \frac{1}{2}\delta_2 \le 1 \right\}\]
The corresponding problem formulation is
\begin{align*}
\min\ & 10x_1 + x_2 + u \\
\text{s.t. }& x_1 + y_1 + x_2 + y_2 = 2 \\
& x_1 + y_1 \le 1 \\
& x_2 + y_2 \le 1 \\
&u \ge y_1 \\
& \frac{1}{2}u \ge y_2 \\
& x_1,x_2 \in\{0,1\} \\
& y_1,y_2 \in [0,1] \\
& u \ge 0
\end{align*}
An optimal solution to this problem is $y_1=x_2=1$ with objective value 2, while an optimal solution to the LP relaxation is $y_1=u=1$ and $x_2=y_2=1/2$ with costs $3/2$.
\end{proof}
Notice that there is still a gap between the 2-approximation algorithm and the integrality gap $4/3$ of the LP relaxation. Closing this gap is an interesting open problem. 

\begin{thm}
The \textsc{RTSt Selection} problem with uncertainty $\cUe$ is approximable within~2.
\end{thm}
\begin{proof}
By solving the relaxation of~(\ref{mip2stsel0})b (which is a continuous, convex optimization problem), we find a solution $(\pmb{x}^*,\pmb{y}^*)$. Using a similar rounding procedure as in the proof of Theorem~\ref{select-2-appr}, we compute a solution $(\hat{\pmb{x}},\hat{\pmb{y}})$ with $2\pmb{y}^*\ge\hat{\pmb{y}}$. As $2\Vert\pmb{A}^T\pmb{y}^*\Vert_2 = \Vert \pmb{A}^T(2\pmb{y}^*)\Vert_2 \ge \Vert \pmb{A}^T\hat{\pmb{y}}\Vert_2$, the approximation guarantee thus follows.
\end{proof}

\begin{thm}
	If the number of budget constraints in $\cUhp_1$ is constant, then \textsc{RTSt Selection} with $\cUhp_1$ admits an FPTAS.
\end{thm}
\begin{proof}
	Using Theorem~\ref{thmfptasgen} it is enough to show that the following problem is polynomially solvable:
	\begin{equation}
	\label{eq0002}
	\begin{array}{lllll}
		\min & \pmb{C}^T\pmb{x}+\underline{\pmb{c}}^T\pmb{y}\\
			\text{s.t.}  &\displaystyle  \sum_{i\in [n]} (x_i+y_i) = p \\
			& \pmb{x}+\pmb{y}\leq \pmb{1} \\
			&  0\leq y_i\leq d_i & i\in [n] \\
			& \pmb{x}\in \{0,1\}^n 
	\end{array}
	\end{equation}
	where $d_i\in \mathcal{E}=\{0,\epsilon, 2\epsilon,\dots,1\}$, $i\in [n]$. We will show first the following property of~(\ref{eq0002}):
	\begin{pro}
		There is an optimal solution to~(\ref{eq0002}) in which $y_i\in \mathcal{E}$ for each $i\in [n]$.
		\label{pxy}
	\end{pro}
	\begin{proof}	
	         Let  $(\pmb{x},\pmb{y})$ be an optimal solution to~(\ref{eq0002}).
		  Since  $\sum_{i\in [n]} (x_i+y_i) = p$,
		   the quantity $\sum_{i\in [n]} y_i=p-\sum_{i\in [n]} x_i$ must be integral. Let us sort the variables so that $\underline{c}_1\leq \underline{c}_2\leq \dots \leq \underline{c}_n$. Let $\ell$ be the first index such that $y_\ell\notin \mathcal{E}$. Notice that $0<y_\ell<d_\ell$. We get $\sum_{i\in [\ell-1]} y_i= k\epsilon$ for some integer $k\geq 0$. Hence $k\epsilon+y_\ell$ cannot be integral and $\sum_{j>\ell} y_j>0$. Set $y_\ell=\min\{\sum_{j>\ell} y_j,d_\ell\}$ and decrease the values of appropriate number of $y_j$, $j>\ell$, so that still $\sum_{i\in [n]} (x_i+y_i)=p$ holds. If $y_\ell=d_\ell$, then we are done as $d_\ell\in \mathcal{E}$. If $y_\ell=\sum_{j>\ell} y_j\leq 1$, then $k\epsilon +y_\ell=p$ and thus $y_\ell\in \mathcal{E}$. Observe that this transformation does not destroy the feasibility of the solution. Furthermore, it also does not increase the solution cost. After applying it a finite number of times we get an optimal solution satisfying the property.
	\end{proof}	
Property~\ref{pxy} allows us to solve~(\ref{eq0002}) by applying a dynamic programming approach. Indeed, 
using the fact that $x_i\in \{0,1\}$ and  $y_i\in \mathcal{E}$ for every $i\in [n]$,
in each stage $i\in [n]$, we have to fix the pair $(x_i, y_i)$, where the feasible assignments are $(0,\epsilon), (0,2\epsilon),\dots, (0,d_i),(1,0)$. A fragment of the computations is shown in Figure~\ref{figdyn}. For each arc we can compute a cost $C_ix_i+\underline{c}_iy_i$. Notice that sometimes there may exist two feasible pairs between two states (see the transition $(s,1)$ in Figure~\ref{figdyn}). In this case, we choose the assignment with smaller cost. 
\begin{figure}[ht]
		\centering
		\includegraphics[height=6cm]{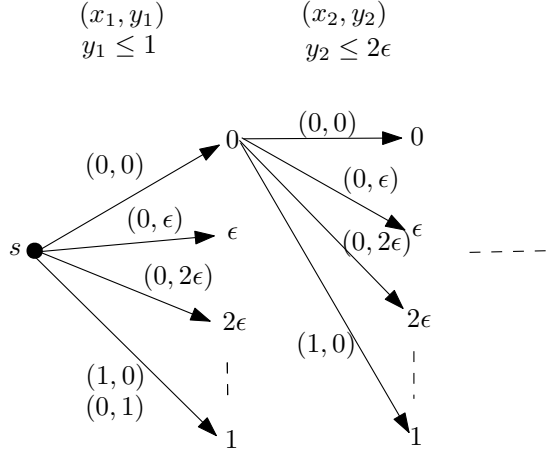}
		\caption{Illustration of the dynamic algorithm.}\label{figdyn}
	\end{figure}

The running time of the dynamic algorithm is $O(np^2\frac{1}{\epsilon^2})$, so it is polynomial when $\epsilon>0$ is fixed.
By Theorem~\ref{thmfptasgen}, the overall running time of the FPTAS is $O(np^2(1/\epsilon)^{K+2})$. 
\end{proof}

\section{Robust two-stage RS problem}
\label{secrs}

In this section we investigate in more detail the robust two-stage version of the \textsc{RS} problem under $\cUhp$ and $\cUe$. In Section~\ref{seccomplex} we proved that this problem is NP-hard.  First observe that for each set $T_l$, $l\in [\ell]$, we have to decide whether to choose a tool in the first or in the second stage. In the former case we always choose the cheapest tool. Hence the problem can be simplified and the MIP formulations~(\ref{mip2st}) and~(\ref{modelip}) for the problem under $\cUhp$ and $\cUe$, respectively,   take the following form
\begin{align}
	 \begin{array}{llllll}
		(a)&\min &\widehat{\pmb{C}}^T\pmb{x}+\underline{\pmb{c}}^T\pmb{y}+\pmb{u}^T\pmb{b}\\
		&	\text{s.t.} & \displaystyle x_l+\sum_{j\in T_l} y_j = 1 & l\in [\ell] \\
		&	& \pmb{u}^T\pmb{A}\geq \pmb{y}^T \\
		&	& \pmb{x}\in \{0,1\}^{\ell} \\
		&	& \pmb{y}, \pmb{u}\geq \pmb{0}
	\end{array}
	&
	\begin{array}{llllllll}
	(b)&\min\ &\widehat{\pmb{C}}^T\pmb{x} + \underline{\pmb{c}}^T\pmb{y} + \Vert \pmb{A}^T\pmb{y}\Vert_2 \\
 &\text{s.t.} & \displaystyle x_l+\sum_{j\in T_l} y_j = 1 & l\in [\ell] \\
&& \pmb{x}\in\{0,1\}^{\ell} \\
&& \pmb{y} \ge 0\\
&&
\end{array}
\label{mip2stsel1}
\end{align}
In the above formulations $\pmb{x}$ is a vector of $\ell$ binary variables corresponding to the tool sets $T_1,\dots, T_{\ell}$, and 
$\widehat{\pmb{C}}=[\widehat{C}_1,\ldots, \widehat{C}_{\ell}]^T$, where
$\widehat{C}_l$,  $l\in [\ell]$, is the smallest first stage cost of the tools in $T_l$, i.e.
$\widehat{C}_l=\min_{j\in T_l}\{C_j\}$.  Note also that there are no 
constraints: $x_l+y_j\leq 1$, $l\in [\ell]$, $j\in T_l$, in (\ref{mip2stsel1}). Now they are  redundant and can be removed. 

\begin{thm}
	The \textsc{RTSt RS} problem under $\cUhp$ and $\cUe$ is approximable within~2.
\end{thm}
\begin{proof}
	Consider an optimal solution $(\pmb{x}^*, \pmb{y}^*, \pmb{u}^*)$ of the LP relaxation of~(\ref{mip2stsel1})a. We form the rounded solution $(\hat{\pmb{x}}, \hat{\pmb{y}}, \hat{\pmb{u}})$ as follows.  For each $l\in [\ell]$, if $x^*_l\geq 0.5$, then we fix $\hat{x}_l=1$ and $\hat{y}_j=0$ for each $j\in T_l$; if $\sum_{j\in T_l} y^*_j\geq 0.5$, then we set $\hat{x}_l=0$ and
	 $\hat{y}_j=y^*_j/\sum_{k\in T_l} y^*_k$
	 for each $j\in T_l$. Obviously in this case $\sum_{j\in T_l} \hat{y}_j=1$ and  $\hat{y}_j\leq 2y^*_j$.
	 We also fix $\hat{u}_i=2u^*_i$ for each $i\in [n]$. Thus the rounded solution is feasible and its cost is at most 2 times the optimum. The same method can be applied to~(\ref{mip2stsel1})b.
\end{proof}
Using the same instance as in the proof of Theorem~\ref{apprguar}, one can show that the worst case ratio of the approximation algorithm is attained.

\begin{thm}
	If the number of budget constraints in $\cUhp_1$ is constant, then \textsc{RTSt RS} with $\cUhp_1$  admits an FPTAS
	\end{thm}
\begin{proof}
	According to Theorem~\ref{thmfptasgen}, it is enough to show that the following problem is polynomially solvable:
	\begin{equation}
	\label{eq0001}
	\begin{array}{lllll}
		\min & \widehat{\pmb{C}}^T\pmb{x}+\underline{\pmb{c}}^T\pmb{y}\\
			\text{s.t.} & x_l+ \displaystyle\sum_{j\in T_l} y_j = 1 & l\in [\ell] \\
			&  0\leq y_j\leq d_j & j\in [n] \\
			& \pmb{x}\in \{0,1\}^{\ell}
	\end{array}
	\end{equation}
	where $d_j\in \mathcal{E}=\{0,\epsilon,2\epsilon,\dots,1\}$, $j\in [n]$.
	We first renumber the variables in each set $T_l$, $l\in [\ell]$, so that they are ordered with respect to nondecreasing  values of $\underline{c}_j$.
For each tool set $T_l$, we greedily allocate the largest possible values to $y_j$, $j\in T_l$, so that the total amount allocated does not exceed~1. If $\sum_{j\in T_l} y_j<1$ or $\widehat{C}_l\leq \sum_{j\in T_l} \underline{c}_jy_j$, then we fix $x_l=1$ and set $y_j=0$ for $j\in T_l$; otherwise we fix $x_i=0$ and keep the allocated values for $y_j$, $j\in T_l$. Using the fact that the variables were initially sorted, the optimal solution can be found in $O(n)$ time. Using Theorem~\ref{thmfptasgen}, we can construct an FPTAS for the problem with running time $O(n\log n+n(1/\epsilon)^K)$.
\end{proof}

\begin{thm}
The \textsc{RTSt RS} problem under $\cUhp_0$ can be solved in $O(n^2\log n)$ time.
\end{thm}
\begin{proof}
 The MIP formulation for the problem under $\cUhp_0$ takes the following form
 (see~(\ref{mip2st})):
$$
\begin{array}{llllll}
\min\ & \widehat{\pmb{C}}^T\pmb{x}+ \underline{\pmb{c}}^T\pmb{y}+ \Gamma \pi + \displaystyle\sum_{j\in[n]}  \rho_j d_j \\
\text{s.t.} &  x_l + \displaystyle\sum_{j\in T_l} y_j = 1 & l\in[\ell] \\
& \pi + \rho_j \ge y_j &  j\in[n] \\
& x_l\in\{0,1\} & l\in[\ell] \\
& y_j \in [0,1] &  j\in [n] \\
&\pi, \pmb{\rho} \ge 0
\end{array}
$$
which can be represented, equivalently, as follows
\begin{equation}
\label{eq007}
\begin{array}{llllll}
\min\ & \widehat{\pmb{C}}^T\pmb{x}+ \underline{\pmb{c}}^T\pmb{y}+ \Gamma \pi + \displaystyle\sum_{j\in[n]} d_j
\max\{0,y_j-\pi\} \\
\text{s.t.} &  x_l + \displaystyle\sum_{j\in T_l} y_j = 1 & l\in[\ell] \\
& x_l\in\{0,1\} & l\in[\ell] \\
& y_j \in [0,1] &  j\in [n] \\
&\pi \in [0,1]
\end{array}
\end{equation}
Substituting $u_j+v_j$ into $y_j$ yields 
\begin{equation}
\label{eq003}
\begin{array}{llllll}
\min & \widehat{\pmb{C}}^T\pmb{x}+ \underline{\pmb{c}}^T\pmb{u} + \Gamma \pi + \displaystyle\sum_{j\in[n]} (\underline{c}_j + d_j) v_j \\
\text{s.t.} &  x_l + \displaystyle\sum_{j\in T_l} (u_j+ v_j) = 1 &l\in[\ell] \\
& x_l\in\{0,1\} &  l\in[\ell] \\
& u_j \in [0,\pi] & j\in [n] \\
& v_j \in [0,1-\pi] & j\in [n] \\
&\pi \in [0,1]
\end{array}
\end{equation}
We now show the following claim:
\begin{cla}\label{claim1}
There is an optimal solution to~(\ref{eq003}) in which $\pi=0$ or  $\pi=\frac{1}{p}$, $p\in [n]$.
\end{cla}
Let an optimal $\pmb{x}$ in~(\ref{eq003}) be fixed and define $S=\{l\in[\ell]: x_l = 0\}$. Define $n' = \sum_{l\in S} |T_l|$ and $m'=|S|$. The optimal values of $\pmb{u}$, $\pmb{v}$ and $\pi$ to~(\ref{eq003}) can then be computed by solving the following LP problem:
\begin{equation}
\begin{array}{lllll}
\min & \displaystyle \sum_{l\in S} \sum_{j\in T_l} \underline{c}_j u_j + \Gamma \pi + \sum_{l\in S}\sum_{j\in T_l} (\underline{c}_j+d_j)v_j \\
 \text{s.t.} &  \displaystyle\sum_{j\in T_l} (u_j + v_j)= 1 & l\in S \\
& u_j + w_j = \pi & j \in T_l, l\in S \\
& v_j + t_j = 1 - \pi &  j \in T_l, l\in S \\
& \pmb{u},\pmb{v},\pmb{w},\pmb{t} \in\mathbb{R}^{n'}_+ \\
& \pi \in [0,1]
\end{array}
\label{eq003f}
\end{equation}
This problem has $4n'+1$ variables, and $2n'+m'$ constraints.
In an optimal basis solution, equivalently \emph{optimal vertex solution}, 
$(\pmb{u},\pmb{v},\pmb{w},\pmb{t},\pi)\in\Rset^{4n'+1}_+$, we therefore have $2n'+m'$ basis and $2n'+1-m'$ non-basis variables.
We start with the following observation,
which is due to the definitions of $u_j$ and $v_j$ in (\ref{eq003}) and the optimality of $(\pmb{u},\pmb{v},\pmb{w},\pmb{t},\pi)$.
\begin{obs}
 If $v_j>0$ then $u_j=\pi$
(if $u_j<\pi$ then $v_j=0$),  $j \in T_l$,  $l\in S$.
\label{obsuv}
\end{obs}

Suppose that $\pi\in (0,1)$. Thus $\pi$ is a basis variable.
Then each constraint of the types $u_j + w_j = \pi$ and $v_j + t_j = 1-\pi$ must contain at least one basis variable apart from $\pi$. There are $m'-1$ of these constraints that have two basis variables apart from $\pi$. Hence at most $m'-1$ variables among $\pmb{u}, \pmb{v}$ have the values different than $\pi$, $1-\pi$, respectively.
Accordingly, there is at least one constraint $l\in S$, such that
 $\sum_{j\in T_{l}}(u_j + v_j) = 1$, where $u_j \in\{0,\pi\}$ and $v_j \in\{0,1-\pi\}$ 
 for every $j\in T_{l}$. Let us denote by $S'$ the set of such constraints, $\emptyset\not=S'\subseteq S$.
 The value of $\sum_{j\in T_{l}}(u_j+ v_j)$ for each $l\in S'$ can be expressed by 
 $p_l\pi+q_l(1-\pi)$, where $p_l$ and $q_l$ are  the numbers of variables  $u_j=\pi$ and $v_j=1-\pi$, respectively,
 in the constraint~$l$.
 Thus
 \begin{equation}
 \sum_{j\in T_{l}}(u_j + v_j)=p_l\pi+q_l(1-\pi)=1, \; l\in S'.
 \label{conl}
 \end{equation}
By Observation~\ref{obsuv}, the form of~(\ref{conl}) and the fact that $\pi\in (0,1)$
one can easily deduce that $p_l\geq 1$ and $0\leq q_l\leq 1$, and
 $q_l=1$ iff $p_l=1$; if $p_l\geq 2$ then $q_l=0$.
 
 Furthermore we claim, for the case $\pi\in (0,1)$, that there always exists at least one constraint~$l'\in S'$, such that
 $\sum_{j\in T_{l'}}(u_j + v_j)=p_{l'}\pi=1$, i.e. $q_{l'}=0$, where $p_{l'}\geq 2$.
 On the contrary, suppose that for each $l\in S'$, the  constraint~$l$ has the form of
  $\sum_{j\in T_{l}}(u_j + v_j)=p_l\pi+q_l(1-\pi)=1$, where $p_l=1$ and $q_l=1$.
  We need to consider two cases.
  The first case $S'=S$. Thus $u_j \in\{0,\pi\}$ and $v_j \in\{0,1-\pi\}$ for every  $j \in T_l$,  $l\in S$,
  and $u_j=\pi$ iff $v_j=1-\pi$. Let us construct a vector
  $\pmb{0}\not=(\pmb{u}^{\epsilon},\pmb{v}^{\epsilon},\pmb{w}^{\epsilon},\pmb{t}^{\epsilon},\pi^{\epsilon})\in\Rset^{4n'+1}$
  as follows:  for every  $j \in T_l$,  $l\in S$,
  set $u^{\epsilon}_j=\epsilon$ and $w^{\epsilon}_j=0$ if $u_j=\pi$;
  $v^{\epsilon}_j=-\epsilon$ and $t^{\epsilon}_j=0$ if $v_j=1-\pi$;
  $u^{\epsilon}_j=0$ and $w^{\epsilon}_j=\epsilon$ if $u_j=0$;
  $v^{\epsilon}_j=0$ and $t^{\epsilon}_j=-\epsilon$ if $v_j=0$;
  and $\pi^{\epsilon}=\epsilon$.
  It is easily seen that 
  $(\pmb{u}-\pmb{u}^{\epsilon},\pmb{u}-\pmb{u}^{\epsilon},\pmb{w}-\pmb{w}^{\epsilon},
  \pmb{t}-\pmb{t}^{\epsilon},\pi-\pi^{\epsilon})$ and
  $(\pmb{u}+\pmb{u}^{\epsilon},\pmb{u}+\pmb{u}^{\epsilon},\pmb{w}+\pmb{w}^{\epsilon},
  \pmb{t}+\pmb{t}^{\epsilon},\pi+\pi^{\epsilon})$ are feasible solutions to~(\ref{eq003f})
  for sufficiently small~$\epsilon>0$. Such $\epsilon$ exists since $\pi\in (0,1)$.
  This contradicts our assumption that $(\pmb{u},\pmb{v},\pmb{w},\pmb{t},\pi)$ is a vertex solution (basis feasible solution).
  The proof for the second case $S'\subset S$ may be handled in much the same way.
  It suffices to notice that for each  
  constraint~$l$,  $\sum_{j\in T_{l}}(u_j + v_j)=1$, $l\in S\setminus S'$, there exits at least one $j'\in  T_{l}$
  such that $0<u_{j'}<\pi$ or  $0<v_{j'}<1-\pi$. Using this fact one can build
  $\pmb{0}\not=(\pmb{u}^{\epsilon},\pmb{v}^{\epsilon},\pmb{w}^{\epsilon},\pmb{t}^{\epsilon},\pi^{\epsilon})\in\Rset^{4n'+1}$
  to arrive to a contradiction with the assumption that $(\pmb{u},\pmb{v},\pmb{w},\pmb{t},\pi)$ is a vertex solution.
  We thus have proved that
   there always exists at least one constraint~$l'\in S'$, such that
 $\sum_{j\in T_{l'}}(u_j + v_j)=p_{l'}\pi=1$, where $p_{l'}\geq 2$.
 Hence $\pi=\frac{1}{p_{l'}}$ for  $\pi\in (0,1)$.
  After adding the boundary values of $\pi$, i.e. $0$ and~$1$, Claim~\ref{claim1} follows.
  

Problem~(\ref{eq003}) can be rewritten as follows:
\begin{equation}
\label{eq004}
\begin{array}{lllllll}
\min & \widehat{\pmb{C}}^T\pmb{x}+ \pi \underline{\pmb{c}}^T\hat{\pmb{u}} + \Gamma \pi + \displaystyle\sum_{j\in[n]} (\underline{c}_j + d_j)(1-\pi) \hat{v}_i \\
 \text{s.t.} & x_l + \sum_{j\in T_l}  (\pi \hat{u}_j + (1-\pi) \hat{v}_j )= 1 &  l\in[\ell] \\
& x_l\in\{0,1\} & l\in[\ell] \\
& \hat{u}_j\in [0,1] & j\in [n] \\
&  \hat{v}_j \in [0,1] &  j\in [n] \\
&\pi \in [0,1]
\end{array}
\end{equation}
where the original variables~$y_j$, $j\in[n]$, in~(\ref{eq007}) are restored as follows: $y_j=\pi \hat{u}_j + (1-\pi) \hat{v}_j$.
Using Claim~\ref{claim1}, let us fix a candidate value for $\pi$. We can now sort with
 respect to nondecreasing values of
the costs $\underline{c}_j$ and $\underline{c}_j+d_j$ of $ \hat{u}_j$ and $ \hat{v}_j$ within each set $T_l$, and either set $x_l=1$ or pack from $\hat{\pmb{u}}$ and
 $\hat{\pmb{v}}$ in nondecreasing order until $\sum_{j\in T_l} \pi \hat{u}_j + (1-\pi)\hat{v}_j$ reaches~1. As there are $O(n)$ values for $\pi$ to check, the overall time required by this method is thus $O(n^2\log n)$.
\end{proof}

\section{Robust two-stage shortest path problem}
\label{secsp}

In Section~\ref{seccomplex} we have shown that \textsc{RTSt Shortest Path} problem is strongly NP-hard even in a very restrictive case, when the cardinality of the set of feasible solutions is~1. We now show that the hardness result can be strengthened.
\begin{thm}
\label{complsp1}
The \textsc{RTSt Shortest Path} problem under $\cUvp={\rm conv}\{\pmb{c}_1,\dots,\pmb{c}_K\}$ is hard to approximate within $\log^{1-\epsilon}K$ for any $\epsilon>0$ unless ${\rm NP}\subseteq {\rm DTIME}(n^{{\rm polylog}\;n})$, even for series-parallel graphs.
\end{thm}
\begin{proof}
Consider the following \textsc{Min-Max Shortest Path} problem. We are given a series-parallel graph $G=(V,A)$, with scenario set $\mathcal{U}=\{\pmb{c}_1,\dots,\pmb{c}_K\}{\subseteq\mathbb{R}^{|A|}_+}$, where scenario $\pmb{c}_j$ is a realization of the arc costs. We seek an $s-t$ path $P$ in $G$ whose maximum cost over $\mathcal{U}$ is minimum. This problem is hard to approximate within $\log^{1-\epsilon}K$ for any $\epsilon>0$ unless ${\rm NP}\subseteq {\rm DTIME}(n^{{\rm polylog}\;n})$~\cite{KZ09}. We construct a cost preserving reduction from \textsc{Min-Max Shortest Path} to \textsc{RTSt Shortest Path} with $\cUvp$. Let us define network $G'=(V',A')$ by splitting each arc $(v_i,v_j)\in A$ into two arcs, namely $(v_i, v_{ij})$ (\emph{dashed arc}) and $(v_{ij}, v_j)$ (\emph{solid arc}). Let $M=|A| c_{\max}+1$, where $c_{\max}$ is the maximal arc cost which appears in $\mathcal{U}$. The first stage costs of all dashed arcs $(v_i, v_{ij})$ are~0 and the first stage costs of all solid arcs $(v_{ij}, v_j)$ are~$M$. For each scenario $\pmb{c}_k\in \mathcal{U}$ we form scenario $\pmb{c}_k'$ under which the costs of dashed arcs $(v_i, v_{ij})$ are~$M$ and the costs of solid arcs $(v_{ij},v_j)$ are equal to the costs of $(v_i, v_j)$ under $\pmb{c}_k$. Finally, we set $\cUvp={\rm conv}\{\pmb{c}_1',\dots,\pmb{c}_K'\}$. Note that $G'$ is series-parallel as well.	
	\begin{figure}[ht]
		\centering
		\includegraphics[height=3.5cm]{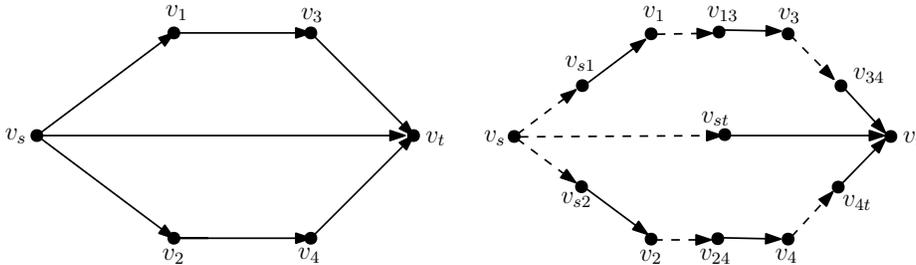}
		\caption{Illustration of the proof of Theorem~\ref{complsp1}.}\label{figc1}
	\end{figure}
	
Observe that only dashed arcs can be selected in the first stage for any partial solution $\pmb{x}$ with $\textsc{Eval}(\pmb{x}) < M$, and only solid arcs can be selected in the second stage. Furthermore if a dashed arc $(v_i, v_{ij})$ is selected in the first stage, then, in order to ensure that a solution built is
an $s-t$ path in~$G'$,  
the solid arc $(v_{ij}, v_j)$ must be selected in the second stage. So, the choice of the arcs in the first stage uniquely gives the set of arcs chosen in the second stage. Let $\pmb{x}$ and $\pmb{y}\in \mathcal{R}(\pmb{x})$ be such a solution to the \textsc{RTSt} problem with total costs less than $M$. The pair $(\pmb{x},\pmb{y})$ is a characteristic vector of an $s-t$ path in $G'$. Since the first stage costs of the dashed arcs are~0, we get
\begin{equation}
\label{eq00}
\textsc{Eval}(\pmb{x})=\max_{\pmb{c}\in \cUvp} \pmb{c}^T \pmb{y}=\max_{\pmb{c}\in\{\pmb{c}_1',\dotsm,\pmb{c}_K'\}} \pmb{c}^T\pmb{y}.
\end{equation}

Suppose there is an $s-t$ path $P=v_s-v_{i_1}-v_{i_2}-\dots-v_t$ in $G$ whose maximum cost over $\mathcal{U}$ is equal to $c$. Path $P$ corresponds to path $P'=v_s-v_{si_1}-v_{i_1}-v_{i_1i_2}-\dots-v_t$ composed of alternated dashed and solid arcs. If $\pmb{x}$ is the characteristic vector of all dashed arcs in $P'$, then $\pmb{y}$ is the characteristic vector of all solid arcs in $P'$. According to~(\ref{eq00}) and the construction of $\pmb{c}'_k$, $k\in [K]$, we have $\textsc{Eval}(\pmb{x})=c$.

Suppose that there is a solution $\pmb{x}$, $\pmb{y}\in \mathcal{R}(\pmb{x})$ to \textsc{RTSt} such that $\textsc{Eval}(\pmb{x})=c$. The characteristic vectors $\pmb{x}, \pmb{y}$ describe a path $P'=v_s-v_{si_1}-v_{i_1}-v_{i_1i_2}-\dots-v_t$ in $G'$ with alternated dashed and solid arcs, where $\pmb{x}$ is the characteristic vector of the dashed arcs and $\pmb{y}$ is the characteristic vector of the solid arcs in $P'$. Using~(\ref{eq00}), we get $\max_{\pmb{c}\in\{\pmb{c}_1',\dotsm,\pmb{c}_K'\}} \pmb{c}^T\pmb{y}=c$. By the construction of the scenarios, we conclude that the maximum cost of the path $P=v_s-v_{i1}-v_{i2}-\dots-v_t$ over $\mathcal{U}$ in $G$ equals $c$.
\end{proof}

Recall that the problem has a $K$-approximation algorithm under $\cUvp$ (see Theorem~\ref{apprK}).

\begin{thm}
\label{complsp2}
The \textsc{RTSt Shortest Path} problem under $\cUhp$ is hard to approximate in graph~$G=(V,A)$
 within $\log^{1-\epsilon}|A|$ for any $\epsilon>0$
unless ${\rm NP}\subseteq {\rm DTIME}(n^{{\rm polylog}\;n})$,
even if $G$ is a series-parallel graph.
\end{thm}	
\begin{proof}
Given an instance of the \textsc{Min-Max Shortest Path} problem with a series parallel graph $G=(V,A)$ and scenario set $\mathcal{U}=\{\pmb{c}_1,\dots,\pmb{c}_K\}{\subseteq\mathbb{R}^{|A|}_+}$, we construct a cost preserving reduction from this problem to \textsc{RTSt Shortest Path} with $\cUhp$. The reduction is similar to the one from the proof of Theorem~\ref{complsp1}. We build a series parallel graph  $G'=(V',A')$ and only add $K$ additional dashed arcs as shown in Figure~\ref{figc2}. These additional dashed arcs have the first stage costs equal to~0 and the second stage costs equal to $M$
($M=|A| c_{\max}+1$, where $c_{\max}$ is the maximal arc cost  in $\mathcal{U}$), so they are all chosen in the first stage.
\begin{figure}[ht]
    \centering
    \includegraphics[height=3.5cm]{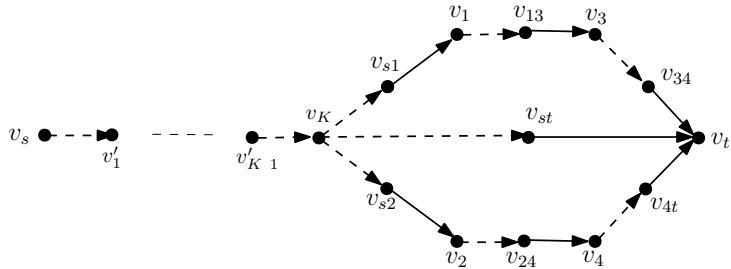}
    \caption{Illustration of the proof of Theorem~\ref{complsp2}.}\label{figc2}
\end{figure}

Define
\[
\cUhp = \left\{ \pmb{0} + \begin{bmatrix}\pmb{\pmb{\delta}}\\ \pmb{\lambda}\end{bmatrix}\ :\;  \pmb{\delta}= \sum_{j\in[K]}  \lambda_i \pmb{c}'_j,
 \sum_{j\in[K]} \lambda_j = 1, 
 \delta_i \ge 0 \ \forall i\in A ,
 \lambda_j \ge 0 \ \forall j\in[K] \right\} \subseteq \mathbb{R}^{2|A|+K}_{+},
\]
where deviations $\pmb{\delta}$ correspond to the arcs of the original graph~$G$ and deviations $\pmb{\lambda}$ correspond to the new dashed arcs. The rest of the proof is similar to the proof of Theorem~\ref{complsp1}. Note that in the hardness proof from \cite[Theorem~1]{KZ09}, we have $|A|\ll K$.
Thus $O(\log |A'|) = O(\log K)$, and the proof is complete.
\end{proof}

\begin{thm}
\label{fptassp}
	If the number of budget constraints in $\cUhp_1$ is constant, then \textsc{RTSt Shortest Path} in  network $G=(V,A)$ with $\cUhp_1$  admits an FPTAS
	\end{thm}
\begin{proof}
	Using Theorem~\ref{thmfptasgen} we need to show that the following problem is polynomially solvable:
	\begin{equation}
	\label{eqsp1}
		 \begin{array}{lllll}
		\min & \pmb{C}^T\pmb{x}+\underline{\pmb{c}}^T\pmb{y}\\
			\text{s.t.} & \displaystyle \sum_{(i,j)\in A} (x_{ij}+y_{ij})-\sum_{(j,i)\in A}(x_{ji}+y_{ji})=\left\{\begin{array}{lll} 1 & i=s \\ -1 & i=t &\\ 0&  i=V\setminus\{s,t\} \end{array} \right. \\
			&0\leq y_{ij}\leq d_{ij} & (i,j)\in A \\
			& \pmb{x}\in \{0,1\}^{|A|}
	\end{array}
	\end{equation}
	where $d_{ij} \in \mathcal{E}=\{0,\epsilon, 2\epsilon,\dots,1\}$, $(i,j)\in A$. 
	
	We will reduce the problem of solving~(\ref{eqsp1}) for fixed~$d_{ij}$ , $(i,j)\in A$,
	to the one of finding a shortest $s-t$ path in an auxiliary directed multigraph~$G'=(V',A')$
	that is built 
	as follows. We first set $V'=V$ and $A'=A$ and associate with each arc $(i,j)\in A'$, the cost
	equal to~$C_{ij}$. We then compute for each  pair of nodes $i\in V$ and $j\in V$, $i\not= j$,
	a cheapest unit flow from $i$ to $j$ in the original graph~$G$
	with respect to the costs $\underline{c}_{ij}$ and 
	arc capacities~$d_{ij}$ and add arc~$(i,j)$ to~$A'$ with the cost equal to the cost of this flow, denoted
	by~$\hat{c}_{ij}$.
	Note that~$\hat{c}_{ij}$ is bounded, if a feasible unit flow  exists, 
	since $\underline{c}_{ij}$ are nonnegative.
	If there is no feasible unit flow from $i$ and $j$, then we do not include~$(i,j)$ to~$A'$.
	The resulting~$G'$ is a multigraph with nonnegative arc costs.
	
	Finally we find a shortest $s-t$ path $P$ in $G'$. We can 
	construct  an optimal solution to~(\ref{eqsp1}) as follows. For each arc $(i,j)\in P$:  if 
	$(i,j)$ has the cost equal to~$C_{ij}$, then set $x_{ij}=1$; otherwise
	 (if $(i,j)$ has the cost equal to~$\hat{c}_{ij}$)
	 fix $y_{ij}$ to the optimal solution of the corresponding min-cost unit flow problem from $i$ to $j$.
	 The rest of variables in~(\ref{eqsp1}) are set to~zero.
	 Since the shortest path and the minimum cost flow problems are polynomially solvable, problem~(\ref{eqsp1}) is polynomially solvable as well. 
	 By Theorem~\ref{thmfptasgen}, the problem admits an FPTAS.
\end{proof}

\section{Conclusions and open problems}

In this paper we have discussed the class of robust two-stage combinatorial optimization problems. We have investigated the general problem as well as several its special cases. The results obtained for the particular problems are summarized in Table~\ref{tabres}. 

\begin{table}[ht]
\centering
\footnotesize
\caption{Summary of the results for the robust two-stage versions of problems $\mathcal{P}$. The symbol \textbf{P} means polynomially solvable.}\label{tabres}
\begin{tabular}{l|llllllllll}
$\mathcal{P}$& $\cUe$ & $\cUvp$ & $\cUhp$ & $\cUhp_0$ & $\cUhp_1$ ($K$-const.) \\ \hline
\textsc{RS} & NP-hard & str. NP-hard & str. NP-hard & \textbf{P} & FPTAS \\
		  & appr. within 2 & appr. within 2 & appr. within 2 \\ \hline
\textsc{Selection} & NP-hard & str. NP-hard & str. NP-hard & \textbf{P}~\cite{CGKZ18} & FPTAS\\
			   & appr. within 2 & appr. within 2 & appr. within 2 \\ \hline 
\textsc{Spanning Tree} & NP-hard & str. NP-hard & str. NP-hard & ?  & ?\\ \hline   
\textsc{Shortest Path} & NP-hard & str. NP-hard & str. NP-hard & ?  & FPTAS \\
				  &		   & appr. within K &  not appr. within& &  \\
				 &		  & not appr. within      &  $\log^{1-\epsilon}|A|$, $\epsilon>0$&   &  \\
				 &		  &$\log^{1-\epsilon}K$, $\epsilon>0$ & 
\end{tabular}
\end{table}

One can see that there is still a number of interesting open questions concerning the robust two-stage approach. The complexity status of the network problem under $\cUhp_0$ is still open. The complexity status of all the problems under $\cUhp_1$, when the number of budget constraints is a part of the input is also open. Also, no positive and negative approximation results have been established for the robust two-stage version of the \textsc{Spanning Tree} problem. For the selection problems, better approximation algorithms can exists. For the ellipsoid uncertainty, we only know that the basic problems are NP-hard. The question whether they are strongly NP-hard and hard to approximate remains open.

\section*{Acknowledgment}
Adam Kasperski and Pawe{\l} Zieli{\'n}ski were  supported by the National Science Centre, Poland, grant 2017/25/B/ST6/00486.


\end{document}